\documentclass[a4paper,UKenglish,cleveref, autoref, thm-restate]{lipics-v2021}

\usepackage{todonotes}
\usepackage{units}
\usepackage{amsmath,stmaryrd,graphicx}
\usepackage[linesnumbered,ruled,vlined]{algorithm2e}
\usepackage{mathtools}
\usepackage{textpos}

\newcommand{\N}{\ensuremath{\mathbb{N}}}

\newcommand{\layer}[1]{\ensuremath{\mathcal{L}_{#1}}}
\newcommand{\lnum}{\ensuremath{\ell}}

\newcommand{\eps}{\varepsilon}

\newcommand{\range}[1]{\ensuremath{[{#1}]}}

\newcommand{\interval}[2]{\ensuremath{[ {#1},{#2} ]}}

\newcommand{\Oh}[2][]{\ensuremath{\mathcal{O}_{#1}({#2})}}

\newcommand{\R}{\ensuremath{\mathbb{R}}}
\newcommand{\Rc}{\ensuremath{\mathcal{R}}}
\newcommand{\Bc}{\ensuremath{\mathcal{B}}}
\newcommand{\Cc}{\ensuremath{\mathcal{C}}}
\newcommand{\Kc}{\ensuremath{\mathcal{K}}}
\newcommand{\Dc}{\ensuremath{\mathcal{D}}}
\newcommand{\Fc}{\ensuremath{\mathcal{F}}}
\newcommand{\Gc}{\ensuremath{\mathcal{G}}}
\newcommand{\Ic}{\ensuremath{\mathcal{I}}}
\newcommand{\Pc}{\ensuremath{\mathcal{P}}}
\newcommand{\Sc}{\ensuremath{\mathcal{S}}}
\newcommand{\Tc}{\ensuremath{\mathcal{T}}}
\newcommand{\Zc}{\ensuremath{\mathcal{S}}}
\newcommand{\matt}{\ensuremath{t}}
\newcommand{\matrixx}{\ensuremath{M}}

\newcommand{\szone}[3]{\ensuremath{\mathsf{subzone}_{#1}(#2,#3)}}
\newcommand{\low}[1]{\ensuremath{\mathsf{low}(#1)}}

\newcommand{\modu}{\ensuremath{\mathsf{mod}}}
\newcommand{\divi}{\ensuremath{\mathsf{div}}}
\newcommand{\object}[1]{\ensuremath{\mathsf{obj}({#1})}}
\newcommand{\ptr}{\ensuremath{\mathsf{ptr}}}
\newcommand{\ptrit}{\ensuremath{\mathsf{ptrIt}}}

\newcommand{\entry}[2]{\ensuremath{\mathsf{entry}({#1},{#2})}}
\newcommand{\ptrglo}{\ensuremath{\mathsf{ptrGlo}}}

\newcommand{\ceil}[1]{\left \lceil #1 \right \rceil }
\newcommand{\floor}[1]{\left \lfloor #1 \right \rfloor }

\makeatletter
\newcommand{\fixed@sra}{$\vrule height 2\fontdimen22\textfont2 width 0pt\shortrightarrow$}
\newcommand{\shortarrow}[1]{%
  \mathrel{\text{\rotatebox[origin=c]{\numexpr#1*45}{\fixed@sra}}}
}
\makeatother

\usepackage{colortbl}
\usepackage{xcolor}

\title{Compact representation for matrices of bounded twin-width} 

\author{Michał Pilipczuk}{Institute of Informatics, University of Warsaw}{michal.pilipczuk@mimuw.edu.pl}{}{}

\author{Marek Sokołowski}{Institute of Informatics, University of Warsaw}{marek.sokolowski@mimuw.edu.pl}{}{}

\author{Anna Zych-Pawlewicz}{Institute of Informatics, University of Warsaw}{anka@mimuw.edu.pl}{}{}
\authorrunning{Mi. Pilipczuk and M. Sokołowski and A. Zych-Pawlewicz} 

\Copyright{Michał Pilipczuk and Marek Sokołowski and Anna Zych-Pawlewicz}



\category{} 

\relatedversion{} 

\funding{This work is a~part of projects that have received funding from the European Research Council (ERC) under the European Union's Horizon 2020 research and innovation programme, grant agreements No. 948057 (Mi.~Pilipczuk, M.~Sokołowski) and 714704 (A.~Zych-Pawlewicz).}

\nolinenumbers 

\begin{document}

\maketitle

\begin{abstract}
For every fixed $d\in \N$, we design a data structure that represents a binary $n\times n$ matrix that is $d$-twin-ordered. The data structure occupies $\Oh[d]{n}$ bits, which is the least one could hope for, and can be queried for entries of the matrix in time $\Oh[d]{\log \log n}$ per query.
\end{abstract}

\begin{textblock}{20}(11.1, -1.2)
	\includegraphics[width=40px]{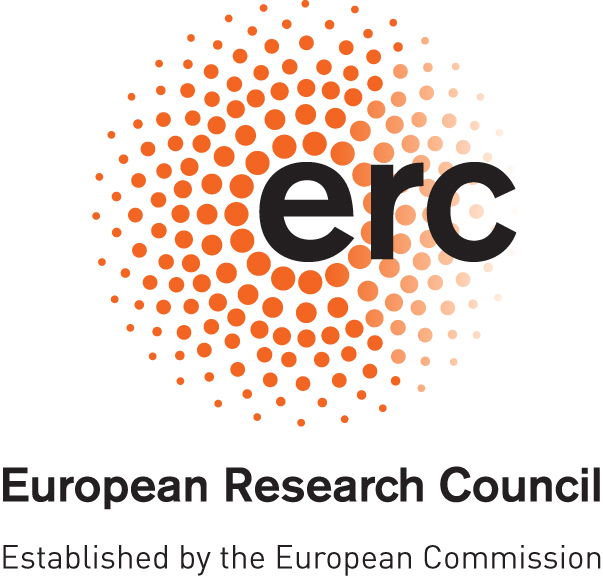}%
\end{textblock}
\begin{textblock}{20}(10.85, -0.8)
	\includegraphics[width=60px]{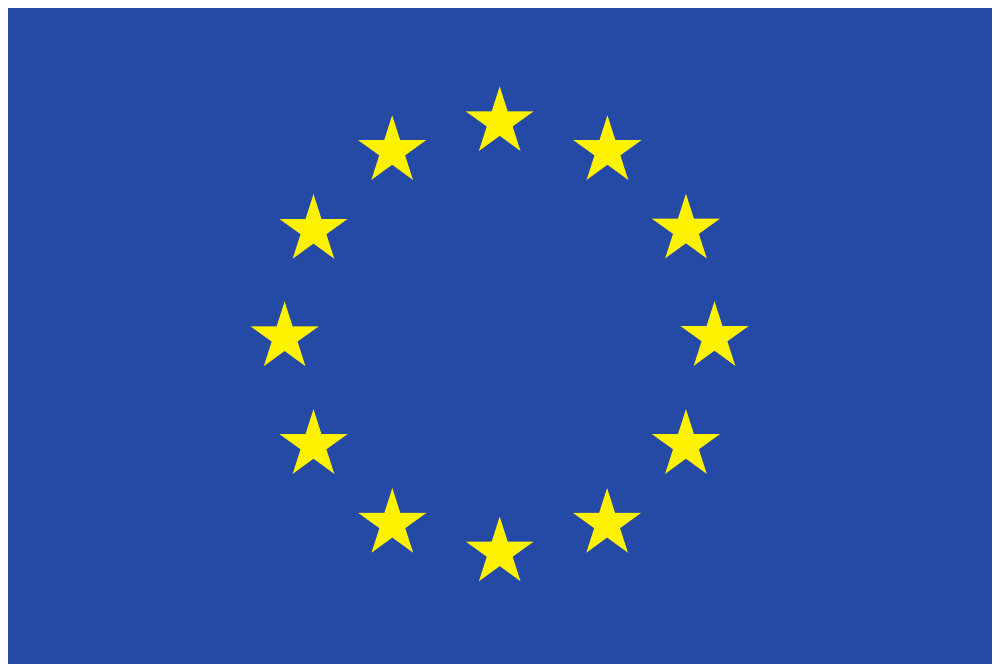}%
\end{textblock}

\newpage

\section{Introduction}\label{sec:intro}

Consider a {\em{binary}} matrix $M$, that is, one with entries in $\{0,1\}$. Given two consecutive columns $c_1,c_2$ of $M$, the operation of {\em{contracting}} those columns consists of replacing them with a single column $c$ with entries reconciled as follows. If in a row $r$ columns $c_1,c_2$ agree, that is, both contain symbol $0$ or both contain symbol $1$, then in column $c$ in row $r$ we put the same symbol. Otherwise, we put a special mismatch symbol $\bot$. Contracting consecutive rows is defined analogously. Contractions of rows and columns can be then applied further with the rule that reconciling any entry with the mismatch symbol $\bot$ again results in $\bot$. For a nonnegative integer $d\in \N$, we say that $M$ is {\em{$d$-twin-ordered}} if $M$ can be contracted to a $1\times 1$ matrix in such a way that at every point during the process, every row and every column contains at most $d$ symbols $\bot$. Finally, the {\em{twin-width}} of $M$ is the least $d$ for which one can permute rows and columns of $M$ so that the resulting matrix is $d$-twin-ordered.

The notion of twin-width was introduced very recently by Bonnet et al. in~\cite{tww1}, and has immediately gathered immense interest. As shown in~\cite{tww1} and in multiple subsequent works~\cite{tww2,tww3,tww4,BonnetKRTW21,BonnetNOdMST21,DreierGJOdMR21,GajarskyPT21}, twin-width is a versatile measure of complexity not only for matrices, but also for permutations and for graphs by considering a suitable matrix representation, which in the latter case is just the adjacency matrix. In particular, for every fixed $t\in \N$, graphs excluding $K_t$ as a minor and graphs having cliquewidth at most $t$ have bounded twin-width, which means that the concept of boundedness of twin-width is a vast generalization of boundedness of cliquewidth that does not assume tree-likeness of the structure of the graph. As shown in the aforementioned works, this generalization is combinatorially rich~\cite{tww2,tww1,DreierGJOdMR21}, algorithmically useful~\cite{tww3,BonnetKRTW21,tww1}, and exposes deep connections with notions studied in finite model theory~\cite{tww4,tww1,BonnetNOdMST21,GajarskyPT21}. In particular, assuming a suitable contraction sequence is provided on input, model-checking First-Order logic on graphs of bounded twin-width can be done in linear fixed-parameter time~\cite{tww1}.

One of the fundamental directions in the work on twin-width is that of estimating the asymptotic {\em{growth}} of considered classes of objects. It has been proved in~\cite{tww2} that the number of distinct graphs on vertex set $\{1,\ldots,n\}$ that have twin-width at most $d$ is bounded by $2^{\Oh[d]{n}}\cdot n!$, which renders the class of graphs of twin-width at most $d$ {\em{small}}. Similarly, the number of distinct $n\times n$ binary matrices that are $d$-twin-ordered is upper bounded by~$2^{\Oh[d]{n}}$~\cite{tww4} (see also the proof of Lemma~\ref{lem:small_family_exp}).

The latter result raises a natural data structure question, which we address in this work. In principle, a binary $n\times n$ matrix of twin-width at most $d$ can be encoded using $\Oh[d]{n}$ bits, just because the number of such matrices is bounded by $2^{\Oh[d]{n}}$. However, would it be possible to design such a representation so that it is algorithmically useful in the following sense: the representation may serve as a data structure that supports efficient queries for entries of the matrix. This question originates from the area of {\em{compact representations}}, see for instance the work on such representations for graphs of bounded cliquewidth~\cite{Kamali18}.

Let us review some solutions to the above problem that to smaller or larger extent follow from existing literature. The quality of a representation is measured by the number of bits it occupies and the worst-case time complexity of a query for an entry. Here, we assume the standard word RAM model with word length $\Oh{\log n}$.
\begin{itemize}
 \item Storing the matrix explicitly is a representation with bitsize $\Oh{n^2}$ and query time $\Oh{1}$. 
 \item In~\cite{tww2}, Bonnet et al. presented an {\em{adjacency labelling scheme}} for graphs of bounded twin-width, which can be readily translated to the matrix setting. This scheme assigns to each row and each column of the matrix a {\em{label}} --- a bitstring of length $\Oh[d]{\log n}$ --- so that the entry in the intersection of a row and a column can be uniquely decoded from the pair of their labels. In~\cite{tww2} the time complexity of this decoding is not analyzed, but a straightforward implementation runs in time linear in the length of labels. This gives a representation with bitsize $\Oh[d]{n\log n}$ and query time $\Oh[d]{\log n}$.
 \item It follows from the results of~\cite{tww3} that if matrix $M$ is $d$-twin-ordered, then the entries $1$ in $M$ can be partitioned into $\ell=\Oh[d]{n}$ rectangles, say $R_1,\ldots,R_\ell$ (see Lemma~\ref{lem:rectangles} for a proof). This reduces our question to {\em{$2$D orthogonal point location}}: designing a data structure that for a given point in $(i,j)\in \{1,\ldots,n\}^2$, may answer whether $(i,j)$ belongs to any of the rectangles $R_1,\ldots,R_\ell$. For this problem, Chan~\cite{DBLP:journals/talg/Chan13} designed a data structure with bitsize $\Oh{n\log n}$ and query time $\Oh{\log \log n}$ assuming $\ell=\Oh{n}$. So we get a representation of $M$ with bitsize $\Oh[d]{n\log n}$ and query time $\Oh[d]{\log \log n}$.
 \item For $2$D orthogonal point location one can also design a simple data structure by persistently recording a sweep of the square $\{1,\ldots,n\}^2$ using a $B$-ary tree for $B=n^{\eps}$, for any fixed $\eps>0$. This gives a representation with bitsize $\Oh[d]{n^{1+\eps}}$ and query time $\Oh{1/\eps}$. See Appendix~\ref{sec:superlinear-repr} for details.
\end{itemize}
Note that in all solutions above, the bitsize of the representation is $\Omega(n\log n)$, and thus does not reach the information-theoretic limit of $\Theta_d(n)$. 

\subparagraph*{Our result.} We design a compact representation for $d$-twin-ordered matrices that simultaneously occupies $\Oh[d]{n}$ bits and offers query time $\Oh[d]{\log \log n}$. The result is summarized in the statement below.

\begin{theorem}\label{thm:main}
 Let $d\in \N$ be a fixed constant.
 Then for a given binary $n\times n$ matrix $M$ that is $d$-twin-ordered one can construct a data structure that occupies $\Oh[d]{n}$ bits and can be queried for entries of $M$ in worst-case time $\Oh{\log \log n}$ per query. The construction time is $\Oh[d]{n\log n\log\log n}$ in the word RAM model, assuming $M$ is given by specifying $\ell=\Oh[d]{n}$ rectangles $R_1,\ldots,R_\ell$ that form a partition of symbols $1$ in $M$. 
\end{theorem}

The proof of Theorem~\ref{thm:main} proceeds roughly as follows. Consider a parameter $m$ that divides $n$ and a partition of the given matrix $M$ into $(n/m)^2$ {\em{zones}} --- square submatrices --- each of which is induced by $m$ consecutive rows and $m$ consecutive columns. Such a partition is called the {\em{ regular $(n/m)$-division}}. Even though the total number of zones in the regular $(n/m)$-division is $(n/m)^2$, one can use the connections between the notions of being twin-ordered and that of {\em{mixed minors}}, developed in~\cite{tww1}, to show that actually there will be only $\Oh[d]{n/m}$ {\em{different}} zones, in the sense that zones are considered equal if they have exactly the same values in corresponding entries. 

Our data structure describes the zones in the regular $(n/m)$-divisions of $M$ for $m$ ranging over a sequence of parameters $m_0>m_1>\ldots>m_\ell$ for $\ell=\Oh{\log\log n}$, where $m_j$ divides $m_i$ whenever $i\leq j$. Roughly speaking, we set $m_0=n$ and $m_{i}=m_{i-1}^{2/3}$ for $i\geq 1$, though for technical reasons we resort to the recursion $m_i=m_{i-1}/2$ once $m_i$ reaches the magnitude of~$\log^3 n$. Each different zone present in the regular $(n/m_{i})$-division is represented by a square matrix consisting of $(m_i/m_{i+1})^2$ pointers to representations of its subzones in the regular $(n/m_{i+1})$-division. When we reach $m_i<c_d\cdot \log n$ for some small constant $c_d$ depending on $d$, we stop the construction and set $\ell=i$. At this point the number of different zones present in the regular $(n/m_\ell)$-division of $M$ is strongly sublinear in $n$, because we have such an upper bound on the total number of different $(c_d\log n)\times (c_d\log n)$ binary matrices that are $d$-twin-ordered, and $n/m_\ell\leq c_d\log n$. Therefore, all those matrices can be stored in the representation explicitly within bitsize $\Oh[d]{n}$.

The query algorithm is very simple: just follow appropriate pointers through the $\Oh{\log \log n}$ levels of the data structure and read the relevant entry in a matrix stored explicitly in the last level. The analysis of bitsize is somewhat more complicated, but crucially relies on the fact that in the $i$th level, it suffices to represent only $\Oh[d]{n/m_i}$ different matrices that are zones in the $(n/m_i)$-division.
%
%
%
%
%

We remark that the idea of dividing the given matrix into a number of polynomially smaller zones, and describing them recursively, is also the cornerstone of the approach used by Chan for the orthogonal point location problem in~\cite{DBLP:journals/talg/Chan13}. However, when it comes to details, his construction is quite different and technically more complicated. For instance, in~\cite{DBLP:journals/talg/Chan13} the recursion can be applied not only on single zones, but also on wide or tall strips consisting of several zones, or even submatrices induced by non-contiguous subsets of rows and columns. The conceptual simplification achieved here comes from the strong properties implied by the assumption that the matrix is $d$-twin-ordered, which is stronger than the assumption used by Chan that the symbols $1$ in the matrix can be partitioned into $\Oh{n}$ rectangles.

\subparagraph*{Organization of the paper.}
In Section~\ref{sec:preliminaries} we define the twin-width of matrices formally and recall a~number of notions related to $d$-twin-ordered matrices.
In Section~\ref{sec:structural-properties}, we prove several new structural properties of those matrices.
These properties are exploited in Section~\ref{sec:main-structure} to derive an~efficient and compact representation of $d$-twin-ordered matrices, completing the non-constructive part of Theorem~\ref{thm:main}.
Finally, the~efficient algorithm for construction of the data structure is given in Section~\ref{sec:oracle-construction}.

\section{Preliminaries}
\label{sec:preliminaries}

For a positive integer $p$ we write $[p]=\{1,\ldots,p\}$. We use the $\Oh[d]{\cdot}$ notation to hide multiplicative factors depending on $d$.


\subparagraph*{Matrices, divisions, and zones.} A {\em{binary matrix}} is a matrix with entries in $\{0,1\}$; all matrices in this paper are binary unless explicitly stated. 

Let $\matrixx$ be a matrix. Note that rows of $\matrixx$ are totally ordered, and similarly for columns of $\matrixx$. A {\em{row block}} in $\matrixx$ is a non-empty set of rows of $\matrixx$ that are consecutive in this total order; {\em{column blocks}} are defined analogously. If $R$ is a row block and $C$ is a column block, then the {\em{zone}} induced by $R$ and $C$ is the rectangular submatrix of $\matrixx$ consisting of entries at the intersections of rows from $R$ and columns from $C$. In general, by a {\em{submatrix}} of $\matrixx$ we mean the zone induced by some row block and some column block.

A submatrix is {\em{constant}} if all its entries are the same. It is horizontal if all its columns are the same (equivalently, all rows are constant), and {\em{vertical}} if all its rows are the same (equivalently, all columns are constant). Note that thus, a constant submatrix is both horizontal and vertical. A submatrix that is neither horizontal nor vertical is called {\em{mixed}}.

A {\em{division}} of matrix $\matrixx$ is a pair $(\Rc,\Cc)$, where $\Rc$ is a partition of rows into row blocks and $\Cc$ is a partition of columns into column blocks. Note that such a division partitions $\matrixx$ into $|\Rc|\cdot |\Cc|$ zones, each induced by a pair of blocks $(R,C)\in \Rc\times \Cc$. We call them the {\em{zones}} of the division $(\Rc,\Cc)$. A \emph{$\matt$-division} is a division where $|\Rc|=|\Cc|=t$. 


\subparagraph*{Twin-width.}
Let $\matrixx$ be a matrix.
If $(\Rc,\Cc)$ is a division of $\matrixx$, then a {\em{contraction}} of $(\Rc,\Cc)$ is any division $(\Rc',\Cc')$ obtained from $(\Rc,\Cc)$ by either merging two consecutive row blocks $R_1,R_2\in \Rc$ into a single row block $R_1\cup R_2$, or merging two consecutive column blocks $C_1,C_2\in \Cc$ into a single column block $C_1\cup C_2$. A {\em{contraction sequence}} for $\matrixx$ is a sequence of divisions 
$$(\Rc_0,\Cc_0),(\Rc_1,\Cc_1),\ldots,(\Rc_p,\Cc_p),$$
such that
\begin{itemize}
 \item $(\Rc_0,\Cc_0)$ is the finest division where every row and every column is in a separate block;
 \item $(\Rc_p,\Cc_p)$ is the coarsest partition where all rows are in a single row block and all columns are in a single column block; and
 \item for each $i\in \{1,\ldots,p\}$, $(\Rc_i,\Cc_i)$ is a contraction of $\Rc_{i-1},\Cc_{i-1}$.
\end{itemize}
Note that thus, $p$ has to be equal to the sum of the dimension of $\matrixx$ minus $2$.

Finally, for a division $(\Rc,\Cc)$ of $\matrixx$, 
the {\em{error value}} of $(\Rc,\Cc)$ is the least $d$ such that in $(\Rc,\Cc)$, every row block and every column block contains at most $d$ non-constant zones.

With all these ingredients in place, we can formally define the twin-width of matrices. There are a few alternative definitions spanning through \cite{tww2,tww3,tww4,tww1}; here we follow the terminology from~\cite{tww1}.

\begin{definition}
 A binary matrix $\matrixx$ is {\em{$d$-twin-ordered}} if it admits a contraction sequence in which every division has error value at most $d$. The {\em{twin-width}} of a binary matrix $\matrixx$ is the least $d$ such that one can permute the rows and columns of $\matrixx$ so that the obtained matrix is $d$-twin-ordered.
\end{definition}

It is straightforward to see that the definition above is equivalent to the one given in the first paragraph of Section~\ref{sec:intro}.

Observe that in the above definition, certifying that a matrix is $d$-twin-ordered requires showing a suitable contraction sequence where all divisions have error value at most $d$. In our algorithmic results we do not require that such a contraction sequence is given on input, as we will exploit the assumption that the matrix is $d$-twin-ordered only through combinatorial properties provided by the connections with mixed minors, which we discuss next. In fact, as discussed~\cite{tww1}, it is currently unknown how to efficiently compute a contraction sequence witnessing that a matrix is $d$-twin-ordered.

\subparagraph*{Matrix minors and Marcus-Tardos Theorem.}
We need the following definitions of matrix minors, which intuitively are ``complicated substructures'' in matrices.

\begin{definition}
 Let $\matrixx$ be a binary matrix. A {\em{$\matt$-grid minor}} in $\matrixx$ is a $\matt$-division of $\matrixx$ where every zone contains at least one entry $1$. A {\em{$\matt$-mixed minor}} in $\matrixx$ is a $\matt$-division of $\matrixx$ where every zone is mixed. We say that $\matrixx$ is {\em{$\matt$-mixed-free}} if $\matrixx$ does not contain a $\matt$-mixed minor.
\end{definition}

The celebrated result of Marcus and Tardos asserts that if a matrix has a large density of entries $1$, then it contains a large grid minor.

\begin{theorem}[\cite{MarcusTardos}]\label{thm:MT}
For every $t\in \N$ there exists $c_t\in \N$ such that the following holds. Suppose $\matrixx$ is an $n\times m$ binary matrix with at least $c_t\cdot \max(n,m)$ entries $1$. Then $\matrixx$ has a $t$-grid minor. 
\end{theorem}

The currently best upper bound on $c_t$ is $\frac{8}{3}(t + 1)^2 2^{4t}$, due to Cibulka and Kyn\v{c}l \cite{cibulka2019better}. From now on we adopt the constant $c_t$ in the notation.

In~\cite{tww1}, Bonnet et al. used the result of Marcus and Tardos to show that, intuitively, large mixed minors are canonical obstacles for having bounded twin-width.

\begin{theorem}[\cite{tww1}]\label{thm:tww1-main}
 Let $\matrixx$ be a binary matrix. Then the following implications hold:
 \begin{itemize}
  \item If $\matrixx$ is $d$-twin-ordered, then $\matrixx$ is $(2d+2)$-mixed-free.
  \item If $\matrixx$ is $t$-mixed-free, then $\matrixx$ has twin-width at most $k_t$, where $k_t$ is a constant depending only on $t$. 
 \end{itemize}
\end{theorem}

Note that the conclusion of the second implication of Theorem~\ref{thm:tww1-main} is only a bound on the twin-width: the matrix might still need to be permuted to be $k_t$-twin-ordered. In this work we will only rely on the first implication of Theorem~\ref{thm:tww1-main}: being $d$-twin-ordered implies $(2d+2)$-mixed-freeness.

We now derive some simple properties of mixed-free matrices that will be used throughout the paper. First, in a $t$-mixed-free matrix every $\ell$-division has only $\Oh[t]{\ell}$ mixed zones.

\begin{lemma}\label{lem:mixedzones}
Let $\matrixx$ be a $t$-mixed free matrix, and let $(\Rc,\Cc)$ be an $\ell$-division of $\matrixx$, for some integer $\ell$. Then $(\Rc,\Cc)$ has at most $c_t \cdot \ell$ mixed zones. 
\end{lemma}
\begin{proof}
Construct an $\ell \times \ell$ matrix $A$ by taking the division $(\Rc,\Cc)$ and substituting each mixed zone with a single entry $1$, and each non-mixed zone with a single entry $0$. Observe that $A$ may have at most $c_t\cdot \ell$ entries $1$, for otherwise, by Theorem \ref{thm:MT}, $A$ would contain a $t$-grid minor, which would correspond to a $t$-mixed minor in $\matrixx$. Hence $(\Rc,\Cc)$ may have at most $c_t\cdot \ell$ mixed zones.
\end{proof}

For the next observations we need the following notion. A {\em{corner}} in a matrix $\matrixx$ is simply a mixed $2\times 2$ submatrix which is an intersection of two consecutive rows with two consecutive columns. The following observation was pivotally used in the proof of Theorem~\ref{thm:tww1-main} in~\cite{tww1}.


\begin{lemma}[\cite{tww1}]
\label{lem:mixed-has-corner}
A matrix is mixed if and only if it contains a corner. 
\end{lemma}

The next lemma is essentially proven in \cite{tww1} but never stated explicitly. So we include a proof for completeness.

\begin{lemma}[implicit in \cite{tww1}]
\label{lem:corner-bound}
A $t$-mixed-free $n\times n$ matrix contains at most $2 c_t (n+2)$ corners.
\end{lemma}
\begin{proof}
Let $\matrixx$ be a $t$-mixed-free $n \times n$ matrix.
Consider the $\ceil{n/2}$-division $(\Rc,\Cc)$ of $\matrixx$, in which every row block consists of rows with indices $2i-1$ and $2i$ for some $i\in \{1,\ldots,\floor{n/2}\}$, possibly except the last block that consists only of row $n$ in case $n$ is odd, and similarly for column blocks. By Lemma~\ref{lem:mixedzones}, $(\Rc,\Cc)$ has at most $c_t\ceil{n/2}\leq c_t(n/2+1)$ mixed zones, which implies that $\matrixx$ has at most $c_t(n/2+1)$ corners in which the bottom-right entry is in the intersection of an even-indexed row and an even-indexed column. Call such corners of {\em{type $00$}}; corners of types $01$, $10$, and $11$ are defined analogously.
By suitably modifying the pairing of rows and columns in $(\Rc,\Cc)$, we can analogously prove that the number of corners of each of the remaining three types is also bounded by $c_t(n/2+1)$. Hence, in total there are at most $4c_t(n/2+1)=2c_t(n+2)$ corners in $\matrixx$.
\end{proof}

Next, we will need a variant of Lemma~\ref{lem:mixedzones} that focuses on mixed borders between neighboring zones. Here, two different zones in a division $(\Rc,\Cc)$ are called {\em{adjacent}} if they are either in the same row block and consecutive column blocks, or in the same column block and consecutive row blocks. A {\em{mixed cut}} in $(\Rc,\Cc)$ is a pair of adjacent zones such that there is a~corner that crosses the boundary between them, i.e., has two entries in each of them. A \emph{split corner} in $(\Rc,\Cc)$ is a corner intersecting four different zones, i.e., it has an entry in four different zones. 

The proof of the following observation is again essentially present in~\cite{tww1}.

%


\begin{lemma}
\label{lem:mixed-cuts}
Let $\matrixx$ be a $t$-mixed free matrix, and let $(\Rc,\Cc)$ be an $\ell$-division of $\matrixx$, for some integer $\ell$. Then $(\Rc,\Cc)$ has at most $c_t\cdot (\ell+2)$ mixed cuts and at most $2 c_t \cdot (\ell+1)$ split corners.
\end{lemma}
\begin{proof}
Let $({\cal R}^{00},{\cal C}^{00})$ be the division obtained from $(\Rc,\Cc)$ by merging the row blocks indexed $2i-1$ and $2i$ into a single row block, and merging the column blocks indexed $2i-1$ and $2i$ into a single column block, for each $i\in \{1,\ldots,\floor{\ell/2}\}$. Obtain divisions $({\cal R}^{10},{\cal C}^{10})$, $({\cal R}^{01},{\cal C}^{01})$, and $({\cal R}^{11},{\cal C}^{11})$ in a similar manner, where if the first number in the superscript is $1$ then we merge row blocks $2i$ and $2i+1$ for each $i\in \{1,\ldots,\ceil{\ell/2}-1\}$ instead, and if the second number in the superscript is $1$ then we merge column blocks $2i$ and $2i+1$ for each $i\in \{1,\ldots,\ceil{\ell/2}-1\}$ instead.

Observe that for every mixed cut of $(\Rc,\Cc)$, the two zones in the mixed cut end up in the same zone in either $({\cal R}^{00},{\cal C}^{00})$ or in $({\cal R}^{11},{\cal C}^{11})$, rendering this zone mixed. However, by Lemma~\ref{lem:mixedzones}, $({\cal R}^{00},{\cal C}^{00})$ and $({\cal R}^{11},{\cal C}^{11})$ have at most $c_t\cdot (\ell/2+1)$ mixed zones. It follows that $(\Rc,\Cc)$ has at most $2c_t\cdot (\ell/2+1)=c_t\cdot (\ell+2)$ mixed cuts. The bound on the number of split corners follows from the same argument combined with the observation that every split corner in $(\Rc,\Cc)$ is entirely contained in a single zone of exactly one of divisions $({\cal R}^{00},{\cal C}^{00})$, $({\cal R}^{10},{\cal C}^{10})$, $({\cal R}^{01},{\cal C}^{01})$, and $({\cal R}^{11},{\cal C}^{11})$.
\end{proof}

\subparagraph*{Partitioning into rectangles.}
We conclude with another observation about twin-ordered matrices: they can be decomposed into a small number of rectangles. Formally, for a binary matrix $\matrixx$, a {\em{rectangle decomposition}} of $\matrixx$ is a set $\Kc$ of pairwise disjoint rectangular submatrices (i.e., zones induced by some row block and some column block) such that every submatrix in $\Kc$ is entirely filled with $1$s and there is no entry $1$ outside the submatrices in $\Kc$. The following lemma is stated and proved in the graph setting in~\cite{tww3}; we adapt the proof here to the matrix setting.

\begin{lemma}\label{lem:rectangles}
 Let $\matrixx$ be an $n\times n$ binary matrix that is $d$-twin-ordered. Then $\matrixx$ admits a rectangle decomposition $\Kc$ with $|\Kc|\leq d(2n-2)+1$.
\end{lemma}
\begin{proof}
 Let $(\Rc_0,\Cc_0),\ldots,(\Rc_{2n-2},\Cc_{2n-2})$ be a contraction sequence for $\matrixx$ with error value at most $d$. Let $\Zc_i$ be the set of zones of $(\Rc_i,\Cc_i)$, and let
 $$\Zc = \bigcup_{i=0}^{2n-2} \Zc_i.$$
 Note that $\Zc$ is a {\em{laminar}} family, that is, every two submatrices in $\Zc$ are either disjoint or one is contained in the other.
 
 Let $\Kc$ be the subfamily of $\Zc$ consisting of those submatrices that are entirely filled with $1$s, and are inclusion-wise maximal in $\Zc$ subject to this property. Note that every entry $1$ in $\matrixx$ is contained in some member of $\Kc$, for the zone of $(\Rc_0,\Cc_0)$ in which this entry is contained is a $1\times 1$ submatrix entirely filled with $1$s.
 Since $\Zc$ is laminar, it follows that $\Kc$ is a rectangle decomposition of $\matrixx$. So it remains to argue that $|\Kc|\leq d(2n-2)+1$.
 
 Consider any $A\in \Kc$ and let $i$ be the largest index such that $A\in \Zc_i$. We may assume that $i<2n-2$, for otherwise the matrix $\matrixx$ is entirely filled with $1$s and the postulated claim is trivial. By maximality, $A$ is contained in a non-constant zone $B\in \Zc_{i+1}$ that resulted from merging $A$ with another adjacent zone $A'\in \Zc_i$, which is not entirely filled with $1$s. In particular, $B$ lies in the unique row block or column block of $(\Rc_{i+1},\Cc_{i+1})$ that resulted from merging two row blocks or two column blocks of $(\Rc_{i},\Cc_{i})$. There can be at most $d$ non-constant zones in this row/column block of $(\Rc_{i+1},\Cc_{i+1})$, and $B$ is one of them. We infer that $i$ can be the largest index satisfying $A\in \Zc_i$ for at most $d$ different submatrices $A\in \Kc$. Since this applies to every index $i\in \{0,1,\ldots,2n-3\}$, we conclude that $|\Kc|\leq d(2n-2)$.
\end{proof}

Observe that Lemma~\ref{lem:rectangles} provides a way to encode an $n\times n$ $d$-twin-ordered matrix in $\Oh[d]{n\log n}$ bits: one only needs to specify the vertices of the submatrices of a rectangle decomposition of size $\Oh[d]{n}$. The proof is also effective, in the sense that given a suitably represented contraction sequence one can compute the obtained decomposition $\Kc$. To abstract away the nuances of representing contraction sequences, throughout this paper we assume that $d$-twin-ordered matrices are provided on input through suitable rectangle decompositions.

\section{Structural properties of divisions}
\label{sec:structural-properties}

Before we proceed to constructing the promised compact representation, we need to describe some new combinatorial properties of twin-ordered matrices.
For the remainder of this section, we fix $d \in \N$ and consider a~matrix $M$ that is $d$-twin-ordered.
In particular, by Theorem~\ref{thm:tww1-main}, $M$ is $(2d+2)$-mixed-free.

\subparagraph*{Strips.}
We begin by considering non-constant vertical and horizontal zones of a~given division of $M$.
We will show that these zones can be grouped into $\Oh[d]{t}$ {\em{strips}} that again are vertical or horizontal, respectively.
This partitioning is formalized as follows.

\begin{definition}
Let $(\Rc,\Cc)$ be a division of a matrix $M$. A~\emph{vertical strip} in $(\Rc,\Cc)$ is an~inclusion-wise maximal set of non-constant vertical zones of $\Dc$ that are contained in the same column block of $(\Rc,\Cc)$, span a contiguous interval of row blocks, and whose union is again a vertical submatrix. {\em{Horizontal strips}} are defined analogously.
\end{definition}

\begin{figure}[h]
  \centering
  \small
  \newcommand{\g}{\cellcolor{blue!12!white}}
  \newcommand{\x}{\cellcolor{blue!26!white}}
  \newcommand{\y}{\cellcolor{yellow!25!white}}
  \newcommand{\z}{\cellcolor{yellow!60!white}}
  \begin{tabular}{ccc|cc|cc|cccc}
  \y 1 & \y 1 & \y 1 & \y 1 & \y 1 & \y 1 & \y 1 & \g 0 & \g 0 & \g 1 & \g 1 \\
  \y 0 & \y 0 & \y 0 & \y 0 & \y 0 & \y 0 & \y 0 & \g 0 & \g 0 & \g 1 & \g 1 \\
  \y 1 & \y 1 & \y 1 & \y 1 & \y 1 & \y 1 & \y 1 & \g 0 & \g 0 & \g 1 & \g 1 \\ \hline
  1 & 1 & 1 & 1 & 1 & 1 & 1 & \g 0 & \g 0 & \g 1 & \g 1 \\ \hline
  0 & 0 & 0 & \z 1 & \z 1 & 1 & 1 & \x 0 & \x 0 & \x 0 & \x 1 \\
  0 & 0 & 0 & \z 0 & \z 0 & 1 & 0 & \x 0 & \x 0 & \x 0 & \x 1 \\ \hline
  \y 1 & \y 1 & \y 1 & \y 1 & \y 1 & 1 & 1 & \x 0 & \x 0 & \x 0 & \x 1 \\
  \y 0 & \y 0 & \y 0 & \y 0 & \y 0 & 1 & 0 & \x 0 & \x 0 & \x 0 & \x 1 \\
  \end{tabular}
  \let\g\undefined
  \let\x\undefined
  \let\y\undefined
  \let\z\undefined
  \caption{Strips in an~example $4$-division of a~matrix.
  Horizontal strips are painted in shades of yellow.
  Vertical strips are painted in shades of blue.
  Unpainted zones are constant or mixed.}
\end{figure}

Naturally, each non-constant vertical zone belongs to exactly one vertical strip; and similarly, each non-constant horizontal zone belongs to exactly one horizontal strip.

We will now show an~upper bound on the number of vertical and horizontal strips present in any $t$-division of $M$.

\begin{lemma}
  \label{lem:cnt_strips_bound}
  For every $t \in \N$, the total number of vertical and horizontal strips in any $t$-division of $M$ is at most $\Oh[d]{t}$.
\end{lemma}
  \begin{proof}
    We focus on the bound for vertical strips only; the proof for horizontal strips is symmetric.
    Fix some $t$-division $(\Rc,\Cc)$ of $M$.
    Observe that each vertical strip $S$ of the division either intersects the top row of the matrix, or the top-most zone of $S$ is adjacent from the top to another zone $C$ such that adding $C$ to $S$ yields a submatrix that is not vertical. (We say that $C$ is adjacent to $S$ {\em{from the top}}.)
    Thus, we partition the family of vertical strips in the $t$-division of $M$ into three types:
    \begin{enumerate}[(I)]
      \item \label{item:lem_cnt_strips_top} strips intersecting the top row of $M$;
      \item \label{item:lem_cnt_strips_mixed} strips adjacent to a~mixed zone $C$ from the top; and
      \item \label{item:lem_cnt_strips_nonmixed} strips adjacent to a~non-mixed zone $C$ from the top.
    \end{enumerate}
    Obviously, there are at most $t$ vertical strips of type (\ref{item:lem_cnt_strips_top}).
    Next, each vertical strip of type (\ref{item:lem_cnt_strips_mixed}) can be assigned a~private mixed zone $C$ adjacent to it from the top.
    Hence, the number of vertical strips of this type is upper bounded by the number of mixed zones in $(\Rc,\Cc)$, which by Lemma~\ref{lem:mixedzones} is bounded by $\Oh[d]{t}$.

    Finally, let us consider vertical strips of type (\ref{item:lem_cnt_strips_nonmixed}).
    Let $S$ be a vertical strip of this type, $D$ be its top-most zone, and $C$ be the non-mixed zone adjacent to $D$ from the top.
    Since $D$ is vertical, all rows of $D$ are repetitions of the same row vector $v_D$. Since $D$ is non-constant, $v_D$ is non-constant as well. 
   
   As $C$ is non-mixed, it is either horizontal or vertical.
   If $C$ is vertical, then all its rows are repetitions of the same row vector $v_C$. Observe that since strip $S$ could not be extended by $C$, we have $v_C\neq v_D$. Now, as $v_D$ is non-constant, it follows that the union of the bottom-most row of $C$ and the top-most row of $D$ contains a corner.
   On the other hand, if $C$ is horizontal, then the bottom-most row of $C$ is constant and again there is a corner in the union of the (constant) bottom-most row of $C$ and the (non-constant) top-most row of~$D$.
   
   So in both cases we conclude that $C$ and $D$ form a mixed cut. By Lemma~\ref{lem:mixed-cuts}, the total number of mixed cuts in $(\Rc,\Cc)$ is bounded by $\Oh[d]{t}$, so also there are at most $\Oh[d]{t}$ vertical strips of type (\ref{item:lem_cnt_strips_nonmixed}).
    This concludes the proof.
  \end{proof}

\subparagraph*{Regular divisions.}
We move our focus to a~central notion of our data structure: \emph{regular divisions} of a~matrix:

\begin{definition}
  Given $M$ and an~integer $s \in \N$, we define the \emph{$s$-regular division} of $M$ as the $\left\lceil \frac{n}{s} \right\rceil$-division of $M$ in which each row block (respectively, column block), possibly except the last one, contains $s$ rows (resp. columns).
  Precisely, if $s \nmid n$, then the last row block and the last column block contain exactly $n \bmod s$ rows or columns, respectively.
\end{definition}

In the data structure, given a~square input matrix $M$, we will construct multiple regular divisions of $M$ of varying granularity (the value of $s$).
Crucially, in order to ensure the space efficiency of the data structure, we will require that the number of \emph{distinct} zones in each such regular division of $M$ should be small.
This is facilitated by the following definition:

\begin{definition}\label{def:zonefamily}
  For $s \in \N$, the \emph{$s$-zone family} of $M$, denoted $\Fc_s(M)$, is the set of all different zones participating in the $s$-regular division of $M$.
\end{definition}

Let us stress that we treat $\Fc_s(M)$ as a set of matrices and do not keep duplicates in it. That is, if the regular $s$-division of $M$ contains two or more isomorphic zones --- with same dimensions and equal corresponding entries --- then these zones are represented in $\Fc_s(M)$ only once.

For the remainder of this section, we will prove good bounds on the cardinality of $\Fc_s(M)$.
Trivially, the cardinality of $\Fc_s(M)$ is bounded by $\left\lceil \frac{n}{s} \right\rceil^2$ (i.e., the number of zones in the $s$-regular division).
Also, the same cardinality is trivially bounded by $2^{\Oh{s^2}}$ (i.e., the total number of distinct matrices with at most $s$ rows and columns).
However, given that $M$ is $d$-twin-ordered, both bounds can be improved dramatically.
First, the dependence on $\frac{n}{s}$ in the former bound can be improved to linear:

\begin{lemma}
  \label{lem:small_family_frac}
  For every $s\in \{1,\ldots,n\}$, the cardinality of $\Fc_s(M)$ is bounded by $\Oh[d]{\frac{n}{s}}$.
\end{lemma}
  \begin{proof}
    First assume that $s \mid n$; hence, each zone in the $s$-regular division of $M$ has $s$ rows and $s$ columns.
    Then, the matrices in $\Fc_s(M)$ can be categorized into four types:
    \begin{itemize}
      \item Constant zones.
        There are at most $2$ of them --- constant $0$ and constant $1$.
      \item Mixed zones.
        Here, Lemma~\ref{lem:mixedzones} applies directly: since the considered division is an $\frac{n}{s}$-division of $M$, there are at most $\Oh[d]{\frac{n}{s}}$ mixed zones in $M$ in total.
      \item Vertical zones.
        By Lemma~\ref{lem:cnt_strips_bound}, all vertical zones of the considered division can be partitioned into $\Oh[d]{\frac{n}{s}}$ vertical strips.
        As all zones have the same dimensions, the zones belonging to a~single vertical strip are pairwise isomorphic.
        From this we infer the $\Oh[d]{\frac{n}{s}}$ upper bound on the number of different vertical zones.
      \item Horizontal zones are handled symmetrically to vertical zones.
    \end{itemize}
    Finally, if $s \nmid n$, then let $M'$ be equal to $M$, truncated to the first $n - (n \bmod s)$ rows and columns; equivalently, $M'$ is equal to $M$ with all zones with fewer than $s$ rows or columns removed.
    The argument given above applies to $M'$, yielding at most $\Oh[d]{\frac{n}{s}}$ different $s \times s$ zones in $M'$ (and equivalently in $M$).
    The proof is concluded by the observation that $M$ contains exactly $2\left\lceil \frac{n}{s} \right\rceil - 1 = \Oh{\frac{n}{s}}$ zones in its $s$-regular division that have fewer than $s$ rows or columns.
  \end{proof}

Second, from the works of Bonnet et al.~\cite{tww2, tww4} one can easily derive an upper bound that is exponential in $s$ rather than in $s^2$:

\begin{lemma}
  \label{lem:small_family_exp}
  For every $s\in \{1,\ldots,n\}$, the cardinality of $\Fc_s(M)$ is bounded by $2^{\Oh[d]{s}}$.
\end{lemma}
  \begin{proof}
    Observe that a submatrix of a $d$-twin-ordered matrix is also $d$-twin-ordered.
    Thus, it is only necessary to upper bound the total number of different $s \times s$ matrices that are $d$-twin-ordered.
    To this end, we use the notion of twin-width of ordered binary relational structures introduced in the work of Bonnet et al.~\cite{tww4}.
    This notion is more general than twin-orderedness in the following sense: each $s \times s$ matrix that is $d$-twin-ordered corresponds to a~different ordered binary structure over $s$ elements of twin-width at most $d$.
    As proved in~\cite{tww4},
    the number of different such structures is upper bounded by $2^{\Oh[d]{s}}$.
    The claim~follows.
  \end{proof}

While the bound postulated by Lemma~\ref{lem:small_family_frac} is more powerful for coarse regular divisions of~$M$ (i.e., $s$-regular divisions for large $s$), Lemma~\ref{lem:small_family_exp} yields a~better bound for $s\leq p_d\cdot \log n$, where $p_d>0$ is a sufficiently small constant depending on $d$.

\section{Data structure}
\label{sec:main-structure}

In this section we present the data structure promised in Theorem~\ref{thm:main}. Recall that it should represent a given binary $n\times n$ matrix $M$ that is $d$-twin-ordered, and it should provide access to the following query: for given $(i,j)\in \range{n}^2$, return the entry $M[i,j]$. Here we focus only on the description of the data structure, implementation of the query, and analysis of the bitsize. The construction algorithm promised in Theorem~\ref{thm:main} is given in Section~\ref{sec:oracle-construction}.

%

Without loss of generality, we assume that $n$ is a~power of $2$. Otherwise we enlarge $\matrixx$, so that its order is the smallest power of $2$ larger than $n$. We use dummy $0$s to fill additional entries. It is straightforward to see that the resulting matrix is $(d+1)$-twin-ordered. Similarly, in the analysis we may assume that $n$ is sufficiently large compared to any constants present in the context.

%

\subparagraph*{Description.}
Our data structure consists of $\lnum+1$ layers: $\layer{0}, \ldots, \layer{\lnum}$. Recall from Definition~\ref{def:zonefamily} that $\Fc_{s}(\matrixx)$ is the family of pairwise different zones participating in the $s$-regular division of $\matrixx$. Each layer $\layer{i}$ in our data structure corresponds to $\Fc_{m_i}(\matrixx)$ for a carefully chosen parameter~$m_i$. Let $\low{x}$ be the largest power of $2$ smaller or equal to $x$. We define parameters $m_i$ inductively as follows: set $m_0=n$ and for $i\geq 0$, 
$$
m_{i+1}=
\begin{cases}
\low{{m_i}^{2/3}} & \text{ if } m_i \geq \log^3 n \\
m_i / 2 & \text{ if } \log n / (2 \beta_d) \leq m_i < \log^3 n 
\end{cases}
$$
where $\beta_d$ is the constant hidden in the $\Oh[d]{\cdot}$ notation in  Lemma~\ref{lem:small_family_exp}, i.e., $|\Fc_s(M)|\leq 2^{\beta_d\cdot s}$.
The construction stops when we reach $m_i$ satisfying $m_i < \log n/(2\beta_d)$, in which case we set $\ell=i$. Note that all parameters $m_i$ are powers of $2$, so $m_j$ divides $m_i$ whenever $i\leq j$.

We also observe the following.

\begin{claim}
 $\lnum\in \Oh{\log \log n}$.
\end{claim}
\begin{claimproof}
 Let $k$ be the least index for which $m_k<\log^3 n$. Observe that for $i\in \range{1,k}$ we have $m_i\leq n^{(2/3)^i}$. So it must be that $k\leq \log_{3/2} \log n+1\in \Oh{\log \log n}$, for otherwise we would have $m_{k-1}\leq n^{(2/3)^{\log_{3/2} \log n}}=n^{1/\log n}=2<\log^3 n$. Next, observe that for $i\in \range{k+1,\lnum}$ we have $m_i=m_k/2^{i-k}$. Therefore, we must have $\ell-k\leq \log (\log^3 n)+1\in \Oh{\log \log n}$, for otherwise we have
 $m_{\ell-1}\leq m_k/2^{\log \log^3 n}<\log^3 n/\log^3 n=1$. The claim follows.
\end{claimproof}


Layer $\layer{\lnum}$ is special and we describe it separately, so let us now describe the content of layer $\layer{i}$ for each $i<\lnum$.  Since $n$ is divisible by $m_i$, every $Z \in \Fc_{m_i}(\matrixx)$ is an $m_i\times m_i$ matrix that appears at least once as a zone in the $(n/m_i)$-regular division of $\matrixx$. Such $Z$ will be represented by an object $\object{Z}$ in $\layer{i}$. Each object $\object{Z}$ stores $(m_i / m_{i+1})^2$ pointers to objects in $\layer{i+1}$; recall here that $m_{i+1}$ divides $m_{i}$.  Consider the $m_{i+1}$-regular division of $Z$. This division consists of $(m_i / m_{i+1})^2$ zones; index them as $\szone{Z}{i}{j}$ for $i,j \in \range{m_i / m_{i+1}}$ naturally. Observe that for all $i,j \in \range{m_i / m_{i+1}}$, it holds that $\szone{Z}{i}{j} \in \Fc_{m_{i+1}}(\matrixx)$. In our data structure, each object $\object{Z} \in \layer{i}$, corresponding to a matrix $Z\in \Fc_{m_i}(\matrixx)$, stores an array $\ptr$ of $(m_i / m_{i+1})^2$ pointers, where $\ptr[i,j]$ points to the address of $\szone{Z}{i}{j}$ for all $i,j \in \range{m_i / m_{i+1}}$. This concludes the description of layer $\layer{i}$ for $i<\lnum$.

We now describe layer $\layer{\lnum}$. It is also a collection of objects, and for each matrix $Z \in \Fc_{m_\lnum}(\matrixx)$ there is an object $\object{Z} \in \layer{\lnum}$; these objects are pointed to by objects from $\layer{\lnum-1}$. However, instead of storing further pointers, each object $\object{Z} \in \layer{\lnum}$ stores the entire matrix $Z \in \Fc_{m_\lnum}(\matrixx)$ as a binary matrix of order $m_\lnum \times m_\lnum$, using $m_\lnum^2$ bits. This concludes the description of $\layer{\lnum}$.

Observe that in $\layer{0}$ there is only one object corresponding to the entire matrix $\matrixx$. We store a global pointer $\ptrglo$ to this object. Our data structure is accessed via $\ptrglo$ upon each query.

\subparagraph*{Implementation of the query.}
The description of the data structure is now complete and we move on to describing how the query is executed. The query is implemented as method $\entry{i}{j}$ and returns $\matrixx[i,j]$; see Algorithm~\ref{alg:query} for the pseudocode (where $\ptrit \rightarrow$ stands for dereference of a pointer $\ptrit$, i.e., the object pointed to by $\ptrit$). Given two integers $i,j \in \range{n}$, the method starts with pointer $\ptrglo$, and uses $i$ and $j$ and iterator pointer $\ptrit$ to navigate via pointers down the layers, ending with a pointer to an object in layer $\layer{\lnum}$. Initially, the iterator $\ptrit$ is set to $\ptrglo$ and it points to $\object{Z}$ for the only matrix $Z \in \Fc_{m_0}(\matrixx)$. Integers $i,j$ are the positions of the desired entry with respect to zone $Z$. 
After a number of iterations, $\ptrit$ points to an object $\object{Z} \in \layer{k}$ for a matrix $Z \in \Fc_{m_k}(\matrixx)$, and maintains current coordinates $i$ and $j$. The invariant is that the desired output is the entry $Z[i,j]$.
In one step of the iteration, the algorithm finds the matrix $Z'$ in $\Fc_{m_{k+1}}(\matrixx)$ containing the desired entry $Z[i,j]$, which is the zone $\szone{Z}{i \; \divi \; m_{k+1}}{j \; \divi \; m_{k+1}}  \in \Fc_{m_{k+1}}(\matrixx)$, and moves the pointer $\ptrit$ to $\object{Z'} \in \layer{k+1}$.  The new coordinates  of the desired entry with respect to $Z'$ are $(i \; \modu \; m_{k+1})$ and $(j \; \modu \; m_{k+1})$, so $i$ and $j$ are altered accordingly. Once the iteration reaches $\layer{\lnum}$, the object pointed to by $\ptrit$ contains the entire zone explicitly, so it suffices to return the desired entry. Obviously, the running time of the query is $\Oh{\log \log n}$, since the algorithm iterates through $\lnum \in \Oh{\log \log n}$ layers. 
\begin{algorithm}
     \SetKwInOut{Input}{Input}
     \SetKwInOut{Output}{Output}
    \SetKw{Raise}{raise exception:\ }
 
     \vskip 0.2cm
     
     \Input{Integers $i,j \in \range{n}$}
     \Output{$\matrixx[i,j]$}
     
     \vskip 0.1cm
     
     $\ptrit \gets \ptrglo$ \\
     
     \For{ $k \gets 0$ \KwTo $\lnum-1$ } 
     {
        $\ptrit \gets (\ptrit \rightarrow \ptr[i \; \divi \; m_{k+1},j \; \divi \; m_{k+1}]$ \;
        $i \gets i \; \modu \; m_{k+1}$ \;
        $j \gets j \; \modu \; m_{k+1}$ \;
     }
        
     \Return $\ptrit \rightarrow Z[i,j]$
     
     \caption{Query algorithm}
     \label{alg:query}
 \end{algorithm}

\subparagraph*{Analysis of bitsize.}
  We now analyze the number of bits occupied by the data structure. First note that the total number of objects stored is bounded by the total number of submatrices of $\matrixx$, which is polynomial in $n$. Hence, every pointer can be represented using $\Oh{\log n}$ bits. Keeping this in mind, the total bitsize occupied by the data structure is proportional to 
  \begin{equation}\label{eq:1}
\sum_{i=0}^{\ell-1} |\Fc_{m_i}(\matrixx)| \left(\frac{m_i}{m_{i+1}}\right)^2 \log n + |\Fc_{m_\ell}(\matrixx)| m_\ell^2,   
  \end{equation}
  This is because for all layers $\layer{i}$ for $i<\ell$ we store $|\Fc_{m_i}(\matrixx)|$ objects, each storing $\left(\frac{m_i}{m_{i+1}}\right)^2$ pointers, and in $\layer{\ell}$ we store $|\Fc_{m_\ell}(\matrixx)|$ objects, each storing a binary matrix of order $m_\ell\times m_\ell$. 
  
We first bound the second term of Equation~(\ref{eq:1}). By Lemma~\ref{lem:small_family_exp}, we have 
$$ |\Fc_{m_\ell}(\matrixx)| m_\ell^2 \leq 2^{\beta_d \cdot m_\ell}\cdot m_\ell^2\leq 2^{\beta_d \cdot \frac{\log n}{2 \beta_d}}\cdot \left(\frac{\log n}{2\beta_d}\right)^2=\sqrt{n}\cdot \left(\frac{\log n}{2 \beta_d}\right)^2\in o(n).$$
We move on to bounding the first term of Equation~(\ref{eq:1}). Let $k$ be the least index for which $m_k < \log^3 n$. We can split the first term of Equation~(\ref{eq:1}) into two sums:
\begin{align}
& \sum_{i=0}^{\ell-1} |\Fc_{m_i}(\matrixx)| \left(\frac{m_i}{m_{i+1}}\right)^2 \log n= \notag \\
& \qquad =\sum_{i=0}^{k-1} |\Fc_{m_i}(\matrixx)| \left(\frac{m_i}{m_{i+1}}\right)^2 \log n + \label{eq:2}\\
& \qquad +\sum_{i=k}^{\ell-1} |\Fc_{m_i}(\matrixx)| \left(\frac{m_i}{m_{i+1}}\right)^2 \log n \label{eq:3}
\end{align}


We first apply Lemma~\ref{lem:small_family_frac} to bound the sum~(\ref{eq:2}). More precisely, if $\alpha_d$ is the constant hidden in the $\Oh[d]{\cdot}$ notation in Lemma~\ref{lem:small_family_frac}, we have
\begin{equation}\label{eq:4}(\ref{eq:2}) \leq \log n\cdot \sum_{i=0}^{k-1} \alpha_d \frac{n}{m_i} \cdot 4 m_i^{2/3}=4 \alpha_d n \log n\cdot \sum_{i=0}^{k-1} \frac{1}{m_i^{1/3}}\end{equation} 
Since for $i\in \range{k-1}$ we have $m_{i+1}=\low{m_i^{2/3}}$ and $m_i\geq \log^3 n$, we have $m_i/m_{i+1}\geq 2$. Therefore $m_i \geq 2^{k-i-1} m_{k-1}$ for $i \in \interval{0}{k-1}$, so we can continue bounding the last expression in Equation~(\ref{eq:4}):
\begin{equation*} (\ref{eq:4}) \leq \alpha_d n \log n \sum_{i=0}^{k-1}\cdot \frac{1}{(2^{k-i-1}m_{k-1})^{1/3}} \leq \alpha_d n\cdot \frac{\log n}{m_{k-1}^{1/3}}\cdot \sum_{i=0}^{k-1} \frac{1}{(2^{k-i-1})^{1/3}} \in \Oh[d]{n}. \end{equation*}
It remains to bound sum~(\ref{eq:3}). We use Lemma~\ref{lem:small_family_frac} similarly as above: $$(\ref{eq:3}) = 4 \log n\cdot \sum_{i=k}^{\ell-1} |\Fc_{m_i}(\matrixx)| \leq  4 \alpha_d n \log n\cdot \sum_{i=k}^{\ell-1} \frac{1}{m_i} \leq 4 \alpha_d n \log n\cdot \sum_{i=0}^{\infty} \frac{1}{(\frac{\log n}{2 \beta_d}) \cdot 2^i}  \in \Oh[d]{n}.$$
By summing up all the bounds we infer that the total number of bits occupied by our data structure is $\Oh[d]{n}$.


\section{Construction algorithm}
\label{sec:oracle-construction}

In this section we complete the proof of Theorem~\ref{thm:main} by presenting an algorithm that constructs the data structure described in  Section~\ref{sec:main-structure} in time $\Oh[d]{n\log n\log\log n}$. Here, we assume that the matrix $\matrixx$ is specified on input by a rectangle decomposition $\Kc$ satisfying $|\Kc|\leq \Oh[d]{n}$. 
%
%
We remark that the our construction algorithm will itself consume superlinear memory.
In fact, even storing the input decomposition $\Kc$ requires $\Omega(n \log n)$ bits of memory.
However, we stress that the data structure constructed by the algorithm occupies only $\Oh[d]{n}$ bits.

The construction will proceed in three phases.
First, in Section~\ref{ssec:submatrix-types}, we will set up a~data structure that given a submatrix $S$ of $M$, returns the \emph{type} of $S$; that is, verifies whether $S$ is constant, vertical, horizontal, or mixed.
Next, in Section~\ref{ssec:construction-meat} we use the results of Section~\ref{ssec:submatrix-types} to find an~effective approximation $\Gc_s(M)$ of the zone families $\Fc_s(M)$.
Finally, these effective approximations will be used in the construction of the data structure itself (Section~\ref{ssec:construction-finale}). 

\newcommand{\ProbSubmatrixTypes}{\textsc{Submatrix Types}\xspace}
\newcommand{\ProbOrthogonalLocation}{\textsc{Orthogonal Point Location}\xspace}
\newcommand{\ProbOrthogonalEmptiness}{\textsc{Orthogonal Range Emptiness}\xspace}
\newcommand{\ProbOrthogonalSegment}{\textsc{Orthogonal Segment Intersection Emptiness}\xspace}
\newcommand{\ProbVerticalRayShooting}{\textsc{Vertical Ray Shooting}\xspace}
\newcommand{\InstOrthogonalLocation}{\ensuremath{\mathcal{I}_L}\xspace}
\newcommand{\InstOrthogonalHorizontal}{\ensuremath{\mathcal{I}_H}\xspace}
\newcommand{\InstOrthogonalVertical}{\ensuremath{\mathcal{I}_V}\xspace}
\newcommand{\InstOrthogonalEmptiness}{\ensuremath{\mathcal{I}_E}\xspace}
\newcommand{\Area}[1]{\ensuremath{\mathcal{A}(#1)}}
\newcommand{\WArea}[1]{\ensuremath{\widehat{\mathcal{A}}(#1)}}
\newcommand{\bnd}[1]{\ensuremath{\partial #1}}

\subsection{Data structure for submatrix types}
\label{ssec:submatrix-types}

We will now define the announced subproblem formally.
In \ProbSubmatrixTypes, we are given a~rectangle decomposition $\Kc$ of an~$n \times n$ matrix $M$, and we are required to preprocess it so as to handle the following queries efficiently: given a~submatrix $S$ of $M$, return:
\begin{itemize}
  \item \emph{constant $c$} ($c \in \{0, 1\}$) if $S$ is constant with $c$ being the common entry;
  \item \emph{horizontal} if $S$ is non-constant horizontal;
  \item \emph{vertical} if $S$ is non-constant vertical; or
  \item \emph{mixed} if $S$ is mixed (i.e., neither horizontal nor vertical).
\end{itemize}
In this section, we prove the following:

\begin{lemma}
  \label{lem:submatrix-types-ds}
  Fix $d \in \N$ and assume $M$ is a binary $n\times n$ matrix that is $d$-twin-ordered.
  Then there is a~data structure for \ProbSubmatrixTypes on $M$ that supports queries in worst-case time $\Oh[d]{\log \log n}$ in the word RAM model.
  The data structure can be constructed in time $\Oh[d]{n \log \log n}$, assuming $M$ is represented on input by a rectangle decomposition $\Kc$ with~$|\Kc|\leq \Oh[d]{n}$.
\end{lemma}

Observe that by restricting $S$ to one-element matrices in Lemma~\ref{lem:submatrix-types-ds}, we will produce a~data structure testing contents of individual entries of $M$ in doubly-logarithmic time --- the same as in the compact representation of $M$ provided Section~\ref{sec:main-structure}.
However, the data structure from Lemma~\ref{lem:submatrix-types-ds} is by no means compact --- in fact, its bitsize is $\Oh{n \log n \log \log n}$, which is even worse than the bitsize $\Oh{n \log n}$ achieved by the direct application of Chan's data structure for orthogonal point location~\cite{DBLP:journals/talg/Chan13}.
Thus, \ProbSubmatrixTypes can only be used as a~building block of an algorithm constructing the compact representation of $M$.

In order to implement the data structure for \ProbSubmatrixTypes, we shall first define three auxiliary geometric problems.
In each problem it can be assumed that each geometric object given on input has integer coordinates between $0$ and $\Oh{n}$.

In \ProbOrthogonalLocation, we are given a~set of $\Oh{n}$ horizontal and vertical segments, where the segments may only intersect at their endpoints.
The segments subdivide the~plane into regions.
In the problem, it is required to preprocess the regions and construct a data structure that can can efficiently locate the~region containing a~given query point. We will use the following data structure of Chan~\cite{DBLP:journals/talg/Chan13} for this problem.

\begin{theorem}[\cite{DBLP:journals/talg/Chan13}]
  \label{thm:efficient-orthogonal-location}
  There is a data structure for \ProbOrthogonalLocation that can answer each query in worst-case time $\Oh{\log \log n}$ and can be constructed in time $\Oh{n \log \log n}$.
\end{theorem}

In \ProbOrthogonalEmptiness, we are given a~set of $\Oh{n}$ points in the plane.
It is required to preprocess the points in order construct a data structure that can efficiently find whether a~queried axis-parallel rectangle contains any of the input points.
In the positive case, it is not required to return any points: a yes/no answer suffices. For this problem, we will use the data structure of Chan et al.~\cite{DBLP:conf/compgeom/ChanLP11}.

\begin{theorem}[\cite{DBLP:conf/compgeom/ChanLP11}]
  \label{thm:efficient-orthogonal-emptiness}
  There is a data structure for \ProbOrthogonalEmptiness that can answer each query in worst-case time $\Oh{\log \log n}$ and can be constructed in time $\Oh{n \log \log n}$.
\end{theorem}

In \ProbOrthogonalSegment, we are given a~set of $\Oh{n}$ horizontal segments in the plane.
It is required to preprocess the segments in order to construct a data structure that can efficiently decide whether a~queried vertical segment intersects any of the horizontal segments.
In the positive case, it is not required to return any segments: a yes/no answer suffices.

\begin{theorem}
  \label{thm:efficient-orthogonal-segments}
  There is a data structure for \ProbOrthogonalSegment that can answer each query in worst-case time $\Oh{\log \log n}$ and can be constructed in time $\Oh{n \log \log n}$.
\end{theorem}
  \begin{proof}
    The problem admits a~trivial reduction to the \ProbVerticalRayShooting problem, in which it is required to preprocess $\Oh{n}$ horizontal segments in order to construct a data structure that can find, for a~given query point $p$, the lowest horizontal segment intersecting the vertical ray shooting upwards from $p$.
    Namely, a~vertical segment $pq$ intersects some horizontal input segment if and only if the lowest horizontal segment returned by an~instance of \ProbVerticalRayShooting for query point $p$ is different than the segment returned for query point $q$.
    As shown by Chan~\cite{DBLP:journals/talg/Chan13}, \ProbVerticalRayShooting is equivalent to \ProbOrthogonalLocation, so we can use Theorem~\ref{thm:efficient-orthogonal-location}.
  \end{proof}

We are now ready to give the data structure for \textsc{Submatrix Types}.

\begin{proof}[Proof of Lemma~\ref{lem:submatrix-types-ds}]
  We interpret $M$ geometrically by representing the matrix as an~$n \times n$ square in the plane, in which each entry corresponds to a~single unit square.
  Let $\Area{M} \coloneqq \bigcup \{[j - 1, j] \times [i - 1, i] \,\mid\, M[i, j] = 1\}$ be the area covered by the $1$ entries of $M$ in this interpretation.
  We remark that $\Area{M}$ is an~orthogonal subset of $[0, n]^2 \subseteq \R^2$.
  Equivalently, $\Area{M}$ can be defined as the (interior-disjoint) union of the rectangles $[j_1 - 1, j_2] \times [i_1 - 1, i_2]$ for each submatrix $M[i_1\ldots i_2, j_1\ldots j_2] \in \Kc$.
  The boundary $\bnd{\Area{M}}$ of $\Area{M}$ can be found in $\Oh[d]{n}$ time by observing that a~unit segment $s$ with integral coordinates is a~subset of $\bnd{\Area{M}}$ if and only if it belongs to the boundary of exactly one rectangle corresponding to a submatrix in $\Kc$.
  For convenience, let $\WArea{M}$ denote the region $\Area{M}$ with all coordinates doubled, and $\bnd{\WArea{M}}$ denote the boundary of $\WArea{M}$.
  
  We now use Theorem~\ref{thm:efficient-orthogonal-location} to set up a data structure \InstOrthogonalLocation for \ProbOrthogonalLocation for $\bnd{\WArea{M}}$.
  Then, given access to \InstOrthogonalLocation, we can verify in $\Oh{\log \log n}$ time whether $M[i, j] = 1$ for given $(i,j) \in \range{n}^2$ by querying \InstOrthogonalLocation whether the point $(2j - 1, 2i - 1)$ belongs to some region that is a~part of $\Area{M}$.
  This, in turn, enables us to locate all corners of $M$.
  Indeed, observe that if $C$ is a corner in $M$, then at least one of the $4$ entries of $C$ is a corner of some submatrix in $\Kc$.
  Hence, by iterating over all submatrices $M[i_1\ldots i_2, j_1\ldots j_2] \in \Kc$ and examining the neighborhood of each of the cells $(i_1, j_1)$, $(i_1, j_2)$, $(i_2, j_1)$, $(i_2, j_2)$ of $M$, we can find all corners in $M$.
  This takes $\Oh[d]{n}$ queries to \InstOrthogonalLocation, and results in a~maximum of $\Oh[d]{n}$ corners in $M$ (reiterating the statement of Lemma~\ref{lem:corner-bound}).
  Henceforth, let $\Bc$ be the set of those pairs $(j,i)$ for which $\{M[i, j],\, M[i, j+1],\, M[i+1, j],\, M[i+1, j+1]\}$ is a~corner in $M$.
  
  Finally, let us consider a~query about the type of a submatrix $S$ of $M$.
  Say that $S=M[r_1\ldots r_2,c_1\ldots c_2]$, that is, $S$ spans the block of rows from $r_1$ to $r_2$, inclusive, and the block of columns from $c_1$ to $c_2$, inclusive ($1 \leq r_1 \leq r_2 \leq n$, $1 \leq c_1 \leq c_2 \leq n$).
  
  We first focus on deciding whether $S$ is mixed or not.
  Recall from Lemma~\ref{lem:mixed-has-corner} that $S$ is mixed if and only if it contains a~corner.
  For this reason, we use Theorem~\ref{thm:efficient-orthogonal-emptiness} to set up a data structure \InstOrthogonalEmptiness for \ProbOrthogonalEmptiness for $\Bc$; this takes time $\Oh[d]{n \log \log n}$.
  Now, $S$ is mixed if and only if $r_1 < r_2$, $c_1 < c_2$, and the rectangle $[c_1, c_2 - 1] \times [r_1, r_2 - 1] \subseteq \R^2$ covers any point in $\Bc$.
  This condition can be verified using \InstOrthogonalEmptiness in time $\Oh[d]{\log \log n}$.
  
  From now on assume that $S$ is not mixed.
  We will now decide whether $S$ is vertical (possibly constant).
  This can be easily done using the following observation:
  
  \begin{claim}
    \label{cl:types-when-vertical}
    Assume that $S$ is not mixed.
    Then $S$ is vertical if and only if the vertical segment $s$ connecting the points $(2c_1 - 1, 2r_1 - 1)$ and $(2c_1 - 1, 2r_2 - 1)$ intersects no horizontal segments of $\bnd{\WArea{M}}$.
  \end{claim}
    \begin{proof}
      ($\Rightarrow$)
      If $S$ is vertical, then $M[r_1, c_1] = M[r_1 + 1, c_1] = \dots = M[r_2, c_1]$.
      Thus, the open rectangle $R \coloneqq (2(c_1 - 1), 2c_1) \times (2(r_1 - 1), 2r_2)$ is either fully contained within $\WArea{M}$ (if $M[r_1, c_1] = 1$), or is disjoint with $\WArea{M}$ (otherwise).
      Hence, $R$ is disjoint with $\bnd{\WArea{M}}$.
      Since $s \subseteq R$, the implication follows.
      
      ($\Leftarrow$) Suppose $S$ is not vertical. Hence, it is horizontal and non-constant, so there exists $r \in \{r_1, r_1 + 1, \dots, r_2 - 1\}$ for which $M[r, c_1] \neq M[r + 1, c_1]$.
      Define now the horizontal segment $m$ connecting $(2(c_1 - 1), r)$ with $(2c_1, r)$.
      By $M[r, c_1] \neq M[r + 1, c_1]$ we have that $m \subseteq \bnd{\WArea{M}}$; thus, $m$ is a~part of some horizontal segment $m'$ of $\bnd{\WArea{M}}$.
      Since $m$ intersects~$s$, so does $m'$.
    \end{proof}
  
  By Claim~\ref{cl:types-when-vertical}, we can determine whether $S$ is vertical as follows. We use Theorem~\ref{thm:efficient-orthogonal-segments} to set up a data structure \InstOrthogonalHorizontal for \ProbOrthogonalSegment for the set of horizontal segments of $\bnd{\WArea{M}}$.
  Since $\bnd{\WArea{M}}$ consists of $\Oh[d]{n}$ segments, \InstOrthogonalHorizontal can be constructed in $\Oh[d]{n \log \log n}$ time.
  Then, verifying whether $S$ is vertical can be reduced to a~single query on \InstOrthogonalHorizontal, which takes $\Oh[d]{\log \log n}$ time.
  Using a~symmetric data structure for vertical segments of $\bnd{\WArea{M}}$, we can also verify whether $S$ is horizontal.
  If $S$ is both vertical and horizontal, then it is constant; in this case, a~single call to \InstOrthogonalLocation is enough to determine whether $S$ is constant $0$ or constant $1$.
  
  Summing up, the construction of the data structure takes $\Oh[d]{n \log \log n}$ time, and each query requires time $\Oh[d]{\log \log n}$ in the worst case. This concludes the proof.
\end{proof}

\subsection{Efficient approximation of zone families}
\label{ssec:construction-meat}

\newcommand{\zone}[3]{\ensuremath{\mathsf{zone}_{#1}({#2}, {#3})}}
\newcommand{\partition}{\ensuremath{\mathcal{U}}\xspace}
\newcommand{\mapping}{\ensuremath{\xi}\xspace}
\newcommand{\GetFirstProc}{\ensuremath{\mathsf{getFirst}()}\xspace}
\newcommand{\ExtendRightProc}{\ensuremath{\mathsf{extendRight}()}\xspace}
\newcommand{\CoverProc}[2]{\ensuremath{\mathsf{cover}(#1, #2)}\xspace}
\newcommand{\ZoneApproxAlgo}{\textsc{Zone Approximation}\xspace}

In this section, we use the findings of Section~\ref{ssec:submatrix-types} to construct a~concise representation of a~given family of zones in the input matrix.
Recall from Lemma~\ref{lem:small_family_frac} that for a~$d$-twin-ordered $n \times n$ matrix $M$, its $s$-zone family $\Fc_s(M)$, defined as the set of distinct zones in the $s$-regular division of $M$, contains at most $\Oh[d]{\frac{n}{s}}$ submatrices.
We shall now generalize this result: given $s \in \N$ and access to $M$ via an oracle for \ProbSubmatrixTypes, we will efficiently compute a~subset $\Gc_s(M)$ of zones of the $s$-regular division of $M$ that {\em{represents}} the $s$-zone family $\Fc_s(M)$ in the following sense: we require that every submatrix in $\Fc_s(M)$ should be represented by at least one zone in $\Gc_s(M)$ equal to this submatrix.
The subset $\Gc_s(M)$ will still contain at most $\Oh[d]{\frac{n}{s}}$ submatrices; hence, it can be regarded as an~efficient over-approximation of $\Fc_s(M)$.
Moreover, we will give an~effective mapping $\xi$, sending any zone of the $s$-regular division of $M$ onto its representative in $\Gc_s(M)$.

Formally, assume that $s \mid n$.
For $i, j \in \left[\frac{n}{s}\right]$, by $\zone{s}{i}{j}$ we mean the zone of the $s$-regular division of $M$ in the intersection of the $i$-th block of rows and the $j$-th block of columns of the division.
Similarly, let $\zone{s}{i_1 \dots i_2}{j_1 \dots j_2} \coloneqq \bigcup_{i = i_1}^{i_2} \bigcup_{j = j_1}^{j_2} \zone{s}{i}{j}$.
We shall prove the following observation:

\begin{lemma}
  \label{lem:efficient-zone-representatives}
  Assume that an~$n \times n$ matrix $M$ is $d$-twin-ordered for a fixed $d \in \N$ and is given through an~oracle $\Tc$ for \ProbSubmatrixTypes from Lemma~\ref{lem:submatrix-types-ds}. Then
  there exists an~algorithm \ZoneApproxAlgo which, given an~integer $s \mid n$, computes:
  \begin{itemize}
    \item a~set $\Gc_s(M) \subseteq \left[\frac{n}{s}\right]^2$ of size $\Oh[d]{\frac{n}{s}}$ and
    \item a~mapping $\mapping_s\,\colon\left[\frac{n}{s}\right]^2 \to \Gc_s(M)$,
  \end{itemize}
   such that for each $i, j \in \left[\frac{n}{s}\right]$, if $(i', j') \coloneqq \mapping(i, j)$ then $\zone{s}{i}{j} = \zone{s}{i'}{j'}$.
  Both $\Gc_s(M)$ and $\mapping_s$ are constructed by \ZoneApproxAlgo in time $\Oh[d]{\frac{n}{s} \log \frac{n}{s} \log \log n}$. For given $(i,j)\in \left[\frac{n}{s}\right]^2$, the value $\mapping_s(i,j)$ can be computed in time $\Oh[d]{\log \log n}$.
\end{lemma}

The remainder of this section is devoted to the proof of Lemma~\ref{lem:efficient-zone-representatives}.

\subparagraph*{Sketch of the algorithm.}
In \ZoneApproxAlgo, we  implement the following strategy. First, create a~partition \partition of the $s$-regular partition of $M$ into $\Oh[d]{\frac{n}{s}}$ contiguous rectangular submatrices, each comprising pairwise equal zones.
Then,  form $\Gc_s(M)$ by picking one zone from each submatrix in \partition.
For the mapping $\mapping_s$, we  set up an~instance of \ProbOrthogonalLocation (Theorem~\ref{thm:efficient-orthogonal-location}).
Given a~query $(i, j)$, we  locate the rectangular submatrix of \partition containing $\zone{s}{i}{j}$, and return the representative of this submatrix.
    
We  consider the following submatrices for \partition:
\begin{itemize}
  \item individual mixed zones;
  \item separate strips (horizontal and vertical); and
  \item constant submatrices of $M$.
\end{itemize}
We now sketch how \partition is populated.
Roughly speaking, the algorithm  traverses all zones $\zone{s}{i}{j}$ of the $s$-regular partition in the row-major order (in the increasing order of $i$, breaking ties in the increasing order of $j$).
The algorithm will repeatedly choose the zone $Z = \zone{s}{i}{j}$ outside of $\bigcup \partition$ that is the earliest in the row-major order.
Then, for some suitably chosen integers $i' \geq i$, $j' \geq j$, a~new submatrix $\zone{s}{i \dots i'}{j \dots j'}$, disjoint with $\bigcup \partition$, will be created and added to \partition.
The new submatrix will have $Z$ in its top-left corner.

Moreover, this process will at each step preserve the following invariant: within each column block of the $s$-partition, \partition covers a~prefix of zones with respect to the row order. Formally, if $\zone{s}{i}{j}$ is a~part of some submatrix of \partition for $i \geq 2$, then so is $\zone{s}{i - 1}{j}$.
Indeed: adding $\zone{s}{i \dots i'}{j \dots j'}$ to \partition would break the invariant only if there existed an~uncovered zone $\zone{s}{\bar{i}}{\bar{j}}$ for some $\bar{i} \in \{1, 2, \dots, i-1\}$, $\bar{j} \in \{j, j+1, \dots, j'\}$.
However, by the choice of $(i, j)$, all such zones already belong to \partition.

\subparagraph*{Auxiliary data structure.}
In order to implement \ZoneApproxAlgo, we first need to show an~efficient way to find the earliest zone in the row-major order that is disjoint with $\bigcup \partition$, under the aforementioned updates of \partition:
    
\begin{lemma}
  \label{lem:rectangle-partition-maintenance}
  Given $m \in \N$, we can construct a~data structure maintaining an~initially empty family \partition of pairwise disjoint subsets of $[m]^2$ under the following queries and updates:
  \begin{itemize}
    \item \GetFirstProc: returns the lexicographically smallest pair of integers $(i, j) \in [m]^2$ outside of $\bigcup \partition$, or $\bot$ if no such pair exists;
    \item \ExtendRightProc: let $(i, j) \coloneqq \GetFirstProc$; the function returns the largest integer $j' \in \{j, j+1, \dots, m\}$ such that all elements $(i, j), (i, j + 1), \dots, (i, j')$ are disjoint with $\bigcup \partition$;
    \item $\CoverProc{i'}{j'}$: let $(i, j) \coloneqq \GetFirstProc$; the function adds a~rectangle $\{i, i + 1, \dots, i'\} \times \{j, j + 1, \dots, j'\}$ as a~new subset of \partition; it is required that $i' \geq i$, $j' \geq j$, and the rectangle is disjoint with $\bigcup \partition$.
  \end{itemize}
  The data structure processes any query in time $\Oh{\log m}$.
  \end{lemma}
  \begin{proof}
    Consider an~array $H[1\dots m]$, where $H[i]$ ($i \in \{1, \dots, m\}$) is defined as the number of distinct $j \in \{1, \dots, m\}$ such that $(i, j) \in \bigcup \partition$.
    By the invariant above, for any pair of integers $i, j \in [m]$, we have that $(i, j) \in \bigcup \partition$ if and only if $j \leq H[i]$.
    The array will be maintained implicitly using a~set $\Sc$ of triples $(h, \ell, r)$ of integers, denoting the maximal intervals of equal values in $H$.
    Formally, $(h, \ell, r) \in \Sc$ if and only if $\ell \leq r$, $H[\ell] = H[\ell + 1] = \dots = H[r] = h$, and $H[\ell - 1] \neq h$, $H[r + 1] \neq h$ (we assume that $H[0] = H[m+1] = \infty$).
    Initially, $\Sc = \{(0, 1, m)\}$.
    The set also maintains the lexicographic order on the triples of integers, as well as there is a linked list that links the elements of $\Sc$ in the natural order in $[1\ldots m]$ (that is, by increasing second, or equivalently third, coordinate).
    Thus, if $\Sc$ is implemented using a~balanced binary search tree, such as an~AVL tree, we can perform any update or query on $\Sc$ in worst-case $\Oh{\log m}$ time.
    
    Given the representation of $H$ through $\Sc$, answering queries \GetFirstProc and \ExtendRightProc is easy in $\Oh{\log m}$ time: let $(h, \ell, r)$ be the lexicographically smallest element of $\Sc$.
    If $h = m$, we return $\bot$; otherwise, $\GetFirstProc = (h + 1, \ell)$ and $\ExtendRightProc = r$.
    Now, consider $\CoverProc{i'}{j'}$ for $i' \geq h+1$, $j' \geq \ell$.
    We must have that $j' \leq r$: by the choice of $(h, \ell, r)$, we know that $H[r + 1] > H[r]$, so the new rectangle cannot extend past the $r$th column of $[m]^2$.
    Hence, we can update $H$ through $\Sc$ by:
    \begin{itemize}
     \item removing $(h, \ell, r)$;
     \item adding $(i', \ell, j')$ back to $\Sc$; and if $j' < r$, also inserting $(h, j'+1, r)$ back to $\Sc$; and
     \item merging $(i', \ell, j')$ with the neighboring intervals in $\Sc$ if necessary, to ensure that $\Sc$ only keeps the maximal intervals of equal values in $H$.
    \end{itemize}
    This involves $\Oh{1}$ updates to $\Sc$.
    Thus, a~single update requires $\Oh{\log m}$ time.
  \end{proof}

  We remark that Lemma~\ref{lem:rectangle-partition-maintenance} implements an~auxiliary method \ExtendRightProc.
  This method is not required to locate the earliest uncovered zone, but will be useful later in the algorithm.
  
  \subparagraph*{Implementation of the algorithm.}
  We now give the implementation of \ZoneApproxAlgo.
  We set up an~instance $\Pc$ of the data structure from Lemma~\ref{lem:rectangle-partition-maintenance} for $m = \frac{n}{s}$.
  In $\Pc$, each element $(i, j) \in [m]^2$ will correspond to the zone $\zone{s}{i}{j}$ of the $s$-regular division.
  Also, recall that $\Tc$ is an~instance of the data structure for \ProbSubmatrixTypes (Lemma~\ref{lem:submatrix-types-ds}) for the matrix $M$.
  We will populate \partition while maintaining the following invariants:
  \begin{enumerate}[(I)]
    \item \label{item:invariant-prefix} within each column block of the $s$-partition, \partition covers a~prefix of zones with respect to the row order; and
    \item \label{item:invariant-strip} each strip of the $s$-partition either is an~element of \partition, or is disjoint with $\bigcup \partition$.
  \end{enumerate}
  Note that we have already shown that Invariant (\ref{item:invariant-prefix}) is preserved throughout the algorithm.
  
  \ZoneApproxAlgo consists of a~main loop which performs the following operations repeatedly: let $(i, j) \coloneqq \GetFirstProc$.
  Depending on the type of $\zone{s}{i}{j}$, we will create a~new rectangular submatrix $S$ of $M$, disjoint with all elements of \partition so far, and add $S$ to \partition by calling $\CoverProc{\cdot}{\cdot}$.
  The loop is repeated until $\GetFirstProc = \bot$.
  
  Thus, assume that some submatrices have been already added to \partition, and let integers $i, j, j_{\max}$ be so that $(i, j) = \GetFirstProc$ and $j_{\max} = \ExtendRightProc$.
  Let $Z \coloneqq \zone{s}{i}{j}$.
  We find the type of $Z$ by a~single call to $\Tc$.
  What we do next depends on the type of the zone:
  \begin{itemize}
    \item \emph{Mixed zone.} In this case, we simply add $Z$ to \partition by calling $\CoverProc{i}{j}$ and proceed to the next iteration of the loop.
    \item \emph{Non-constant vertical zone.}
    We perform a~binary search to locate the largest integer $i' \in \{i, i + 1, \dots, \frac{n}{s}\}$ for which the submatrix $Z' \coloneqq \zone{s}{i\dots i'}{j \dots j}$ is vertical.
    This requires $\Oh{\log\frac{n}{s}}$ calls to $\Tc$.
    Then we add $Z'$ to \partition by calling $\CoverProc{i'}{j}$.
    \item \emph{Non-constant horizontal zone.}
    As above, use binary search to find the largest index $j' \in \{j, j + 1, \dots, j_{\max}\}$ for which the submatrix $Z' \coloneqq \zone{s}{i \dots i}{j \dots j'}$ is horizontal.
    We add $Z'$ to \partition by calling $\CoverProc{i}{j'}$.
    \item \emph{Constant zone.}
    We use the same binary search as in the vertical case to locate the largest index $i' \in \{i, i+1, \dots, \frac{n}{s}\}$ such that the submatrix $Z' \coloneqq \zone{s}{i \dots i'}{j \dots j}$ is constant.
    We then run another binary search to find the largest index $j' \in \{j, j+1, \dots, j_{\max}\}$ such that the submatrix $Z'' \coloneqq \zone{s}{i \dots i'}{j \dots j'}$ is constant.
    Then, we add $Z''$ to \partition by calling $\CoverProc{i'}{j'}$.
  \end{itemize}
  
  It is easy to see that Invariants~(\ref{item:invariant-prefix}) and (\ref{item:invariant-strip}) ensure that the new submatrix is disjoint with $\bigcup \partition$.
  Then, Invariant~(\ref{item:invariant-strip}) guarantees that the zone $Z$ in the vertical and horizontal cases is the earliest zone in the row-major order of the strip $S$ containing $Z$, and $S$ is disjoint with $\bigcup \partition$.
  Thus, by the definition of a~strip as a~maximal vertical of horizontal submatrix, the presented binary search scheme will find the submatrix $Z'$ equal to $S$.
  Adding the strip to \partition maintains the invariant.
  
  After the main loop terminates, \partition is a~partitioning of $M$ into rectangular submatrices of $M$: mixed zones, strips, and a~number of constant submatrices.
  For each submatrix, we locate its earliest zone $\zone{s}{i}{j}$ in the row-major order, and we add $(i, j)$ to $\Gc_s(M)$.
  Thus, $|\Gc_s(M)| = |\partition|$.
  For $\mapping_s$, observe that \partition is isomorphic to a~subdivision of the square $[0, \frac{n}{s}]^2 \subseteq \R^2$ into rectangular regions, each corresponding to a~single submatrix of \partition.
  Thus, we instantiate an~instance $\Ic_L$ of \ProbOrthogonalLocation for this set of rectangles.
  Each query $\mapping_s(i, j)$ is relayed to $\Ic_L$.
  The answer from $\Ic_L$ can be translated into a~reference to the rectangular submatrix $S$ of $M$ containing $\zone{s}{i}{j}$.
  The value of $\mapping_s(i, j)$ can then be immediately deduced from $S$.
  
  \subparagraph*{Analysis of the algorithm.}
  First, we bound the number of iterations of the main loop:
  
  \begin{lemma}
    \label{lem:small-partition}
    $|\partition| \leq \Oh[d]{\frac{n}{s}}$.
  \end{lemma}

  Before we prove Lemma~\ref{lem:small-partition}, let us verify that the time complexity of the algorithm \ZoneApproxAlgo promised in the statement of Lemma~\ref{lem:efficient-zone-representatives} follows from it.
  The main loop of the algorithm runs $\Oh{|\partition|} = \Oh[d]{\frac{n}{s}}$ times.
  Therefore, the oracle $\Tc$ is called $\Oh[d]{\frac{n}{s} \log \frac{n}{s}}$ times in our algorithm, requiring $\Oh[d]{\frac{n}{s} \log \frac{n}{s} \log \log n}$ time in total.
  Next, the time complexity of all calls to $\Pc$ is bounded by $\Oh[d]{\frac{n}{s} \log \frac{n}{s}}$.
  Finally, the time complexity of the construction of $\Ic_L$ is bounded by $\Oh[d]{\frac{n}{s} \log \log \frac{n}{s}}$, and each call to $\mapping_s$ takes time $\Oh{\log \log \frac{n}{s}}$.
  Thus, the time complexity analysis of \ZoneApproxAlgo is complete; it remains to prove Lemma~\ref{lem:small-partition}.
  
    \begin{proof}[Proof of Lemma~\ref{lem:small-partition}]
      In the regular $s$-division of $M$, there are at most $\Oh[d]{\frac{n}{s}}$ mixed zones (Lemma~\ref{lem:mixedzones}) and at most $\Oh[d]{\frac{n}{s}}$ strips (Lemma~\ref{lem:cnt_strips_bound}).
      It remains to bound the number of constant submatrices in $\partition$ by $\Oh[d]{\frac{n}{s}}$.
      
      We say that a~constant submatrix $S = \zone{s}{i_1 \dots i_2}{j_1 \dots j_2}$ is \emph{guarded} if $S$ either touches the boundary of $M$ (i.e., $i_1 = 1$, $j_1 = 1$, $i_2 = \frac{n}{s}$, or $j_2 = \frac{n}{s}$), or the slightly larger submatrix $\zone{s}{i_1-1 \dots i_2+1}{\ j_1 - 1 \dots j_2 + 1}$, called the \emph{shell} of $S$, is mixed.
      
      \begin{claim}
      \label{cl:each_submatrix_guarded}
      Every constant submatrix in \partition is guarded.
    \end{claim}
      \begin{proof}
      Assume otherwise, and choose an~unguarded constant submatrix $S \in \partition$.
      Then, $S$ is not incident to the boundary of $M$, and each zone adjacent to $S$ (by a~side or by a~corner) is horizontal or vertical.

      Suppose $S = \zone{s}{i_1 \dots i_2}{j_1 \dots j_2}$.
      Let $\widehat{S} \coloneqq \zone{s}{i_1-1 \dots i_2+1}{\ j_1 - 1 \dots j_2 + 1}$ be the shell of $S$.
      Without loss of generality, let $0$ be the common entry in $S$.
    
      Let $S_{\shortarrow{3}}$ be the zone in the top left corner of $S$ (i.e., the earliest zone of $S$ in the row-major order), and $S_{\shortarrow{5}}$ be the zone in the bottom left corner of $S$.
      We also consider the following zones in $\widehat{S}$: $Y_{\shortarrow{2}}$, $Y_{\shortarrow{3}}$, and $Y_{\shortarrow{4}}$, adjacent to $S_{\shortarrow{3}}$ from the top, top left, and left, respectively; and $Y_{\shortarrow{6}}$ and $Y_{\shortarrow{5}}$, adjacent to $S_{\shortarrow{5}}$ from the bottom and bottom left, respectively (Figure~\ref{fig:guarded_submatrix_config}).
    
      \begin{figure}[h]
      \centering
      \begin{tikzpicture}
      \draw [black] (0, 0) rectangle (6, 4);
      \draw [fill=gray!15!white, thick, dashed] (0.8, 0.3) rectangle (4.2, 3.7);
      \draw [fill=gray!50!white, thick] (1.5, 1.0) rectangle (3.7, 3.0);
      \foreach \x in {0.8, 1.5, 2.6, 3.7, 4.2, 5.0} {
        \draw [gray!90!black] (\x, 0) -- (\x, 4);
      };
      \foreach \y in {0.3, 1.0, 1.8, 2.4, 3.0, 3.7} {
        \draw [gray!90!black] (0, \y) -- (6, \y);
      }
      \node at (2.05, 1.4) {$S_{\shortarrow{5}}$};
      \node at (2.05, 2.7) {$S_{\shortarrow{3}}$};
      \node at (2.05, 3.35) {$Y_{\shortarrow{2}}$};
      \node at (1.15, 3.35) {$Y_{\shortarrow{3}}$};
      \node at (1.15, 2.7) {$Y_{\shortarrow{4}}$};
      \node at (1.15, 0.65) {$Y_{\shortarrow{5}}$};
      \node at (2.05, 0.65) {$Y_{\shortarrow{6}}$};
      \end{tikzpicture}
      \caption{An~example constant submatrix $S$ (dark gray). The zones in $\widehat{S} \setminus S$ are marked light gray. In this figure, the first row block is at the top, and the first column block is on the left.}
      \label{fig:guarded_submatrix_config}
      \end{figure}
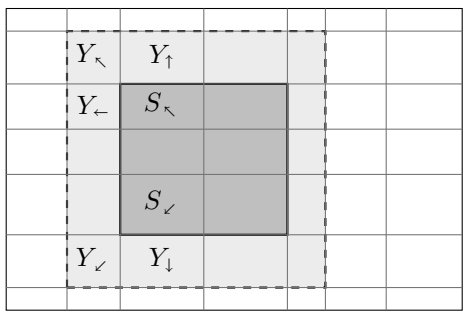
      
      Consider the zone $Y_{\shortarrow{6}}$.
      As $\widehat{S}$ is not mixed, $Y_{\shortarrow{6}}$ is not mixed either.
      Moreover, by the construction of $S$, the submatrix $\zone{s}{i_1 \dots i_2 + 1}{j_1 \dots j_1}$ is not constant, so the zone $Y_{\shortarrow{6}}$ is not constant $0$.
      Also, $Y_{\shortarrow{6}}$ cannot be non-constant vertical; otherwise, a~mixed cut would appear on the boundary between $S_{\shortarrow{5}}$ and $Y_{\shortarrow{6}}$, and $\widehat{S}$ would be mixed.
      Hence, $Y_{\shortarrow{6}}$ is either non-constant horizontal or constant $1$.
      It can be now easily verified that $\widehat{S}$ is non-constant horizontal.
      It immediately follows that: $Y_{\shortarrow{4}}$ is constant $0$; both $Y_{\shortarrow{5}}$ and $Y_{\shortarrow{6}}$ are non-constant horizontal; and both $Y_{\shortarrow{3}}$ and $Y_{\shortarrow{2}}$ are horizontal (possibly constant) and repeat the same column vector.
    
      Since the elements of $\partition$ are rectangular, either $Y_{\shortarrow{2}}$ or $Y_{\shortarrow{4}}$ must belong to a~different element of $\partition$ than $Y_{\shortarrow{3}}$.
      We proceed by refuting both cases.
    
      \smallskip
    
      \emph{\underline{Case 1:} $Y_{\shortarrow{4}}$ is in a~different submatrix $A \in \partition$ than $Y_{\shortarrow{3}}$.}
      Then, $Y_{\shortarrow{4}}$ is the zone in the top-right corner of $A$.
      Also, $A$ is constant $0$, since $Y_{\shortarrow{4}}$ is constant $0$.
      Moreover, $Y_{\shortarrow{5}}$ is not constant and thus remains outside of $A$.
      Hence, the set of rows spanned by $A$ is a~subset of the set of rows spanned by $S$.
      But the top-left zone of $A$ is earlier in the row-major order than the top-left zone of $S$, so $A$ must have been added to \partition by the algorithm before $S$; and when $A$ was being added to \partition, all zones of $S$ were outside of \partition.
      Hence, the binary search scheme would have extended $A$ more to the right, in particular covering $S_{\shortarrow{3}}$ as a~zone --- a~contradiction since $S_{\shortarrow{3}} \notin A$.
    
      \smallskip
    
      \emph{\underline{Case 2:} $Y_{\shortarrow{2}}$ is in a~different submatrix $A \in \partition$ than $Y_{\shortarrow{3}}$.}
      Then, $Y_{\shortarrow{2}}$ is constant; otherwise, it would be non-constant horizontal, and would form a~strip together with $Y_{\shortarrow{3}}$.
      Hence, $Y_{\shortarrow{2}}$ is the zone in the bottom-left corner of $A$.
      But again, $A$ would have been added to \partition by the algorithm before $S$, and $A$ would have extended downwards to cover $S_{\shortarrow{3}}$ --- a~contradiction.
    
    \smallskip
    
      Since all cases have been exhausted, we conclude that $S$ must be guarded.
    \end{proof}
    
    \begin{claim}
    \label{cl:few-constant-submatrices}
    \partition contains at most $\Oh[d]{\frac{n}{s}}$ constant submatrices.
    \end{claim}
    \begin{proof}
    Obviously, there are at most $4\cdot\frac{n}{s}$ different submatrices of \partition touching the boundary of $M$.
    Consider then a~constant submatrix $S \in \partition$ that does not touch the boundary of $M$.
    Let $S = \zone{s}{i_1\dots i_2}{j_1 \dots j_2}$, and let $\widehat{S} \coloneqq \zone{s}{i_1-1 \dots i_2+1}{\ j_1 - 1 \dots j_2 + 1}$ be the shell of $S$.
    By Claim~\ref{cl:each_submatrix_guarded}, $\widehat{S}$ is mixed, so it contains a~corner $C$.
    We consider three cases, depending on the location of $C$.
    \begin{itemize}
      \item If $C$ is fully contained within some (mixed) zone $Z$, we assign $S$ to $Z$.
      \item If $C$ is split in halves by some (mixed) cut $\mu$, we assign $S$ to $\mu$.
      \item If $C$ is split by the zone boundaries into four $1 \times 1$ submatrices, we assign $S$ to $C$.
    \end{itemize}
    As all submatrices of \partition are pairwise disjoint, each entry of $M$ belongs to at most $9$ shells of the submatrices of \partition.
    In particular, each object (mixed zone, mixed cut or corner) belongs to at most $9$ shells.
    It follows that each such object may be assigned to at most $\Oh{1}$ guarded submatrices.
    Since the $s$-regular division of $M$ contains at most $\Oh[d]{\frac{n}{s}}$ mixed zones (Lemma~\ref{lem:mixedzones}), $\Oh[d]{\frac{n}{s}}$ mixed cuts (Lemma~\ref{lem:mixed-cuts}), and $\Oh[d]{\frac{n}{s}}$ split  corners (also Lemma~\ref{lem:mixed-cuts}), we conclude that \partition contains at most $\Oh[d]{\frac{n}{s}}$ constant submatrices.
    \end{proof}
    As discussed, with Claim~\ref{cl:few-constant-submatrices} established, the statement of the lemma is immediate.
    \end{proof}

\subsection{Construction algorithm for Theorem~\ref{thm:main}}
\label{ssec:construction-finale}

\newcommand{\filtering}{\ensuremath{\psi}}

We now combine the results of Sections~\ref{ssec:submatrix-types} and~\ref{ssec:construction-meat} to construct the data structure described in Section~\ref{sec:main-structure}.
As promised, the construction will take time $\Oh[d]{n \log n \log \log n}$, provided that the input matrix is given by specifying a~rectangle decomposition $\Kc$ with~$|\Kc|\leq \Oh[d]{n}$.
This will conclude the proof of Theorem~\ref{thm:main}.

First, we set up an~instance $\Ic_T$ of the data structure for \ProbSubmatrixTypes on $M$.
$\Ic_T$~can be instantiated in time $\Oh[d]{n \log \log n}$ from $\Kc$ (Lemma~\ref{lem:submatrix-types-ds}).
Recall that each access to $\Ic_T$ is realized in worst-case time $\Oh{\log \log n}$.
From now on, we assume that $M$ is accessed only through $\Ic_T$.

\subparagraph*{Recap of the data structure.}
In Section~\ref{sec:main-structure} we defined parameters $m_0 > m_1 > \dots > m_\ell$ such that $m_0 = n$, $m_\ell \in \Theta_d(\log n)$, $\ell \in \Oh{\log \log n}$, and $m_{i+1} \mid m_i$ and $m_{i+1}^2 \geq m_i$ for each $i \in \{0, 1, \dots, \ell - 1\}$.
The data structure consists of $\ell + 1$ layers: $\layer{0}, \dots, \layer{\ell}$.
Each layer $\layer{i}$ contains one object $\object{Z}$ for each element $Z$ of the zone family $\Fc_{m_i}$.
For $i < \ell$, each object $\object{Z}$ in $\layer{i}$ contains an~$\frac{m_i}{m_{i+1}} \times \frac{m_i}{m_{i+1}}$ array $\ptr$ of pointers to the objects in $\layer{i+1}$, corresponding to the elements of the regular $m_{i+1}$-division of $Z$.
For $i = \ell$, each object $\object{Z}$ in $\layer{\ell}$ stores the entire submatrix $Z$ using $m_\ell^2$ bits.
By carefully choosing the parameters, we guarantee that the data structure occupies $\Oh[d]{n}$ bits.

\subparagraph*{Strategy.}
Assume that the suitable parameters $m_0, \dots, m_\ell$ have already been selected.
We construct the data structure bottom-up, starting from layer $\layer{\ell}$, and concluding with layer~$\layer{0}$.
For each $i = \ell, \ell - 1, \dots, 0$, we construct a~set $\Gc_{m_i}(M)$ of representative zones in the regular $m_i$-division of $M$; and a~mapping $\mapping_{m_i}\,\colon\left[\frac{n}{m_i}\right]^2 \to \Gc_{m_i}(M)$, sending any zone of the $m_i$-regular division of $M$ onto their representative in $\Gc_{m_i}(M)$.
Since $m_i \in \Omega_d(\log n)$, this construction will take time $\Oh[d]{n \log \log n}$ for each $i$.
Moreover, $|\Gc_{m_i}(M)| \in \Oh[d]{\frac{n}{m_i}}$, and the mapping $\mapping_{m_i}$ can be evaluated on any zone in $\Oh{\log \log n}$ time.
Next, we construct $\Fc_{m_i}(M)$ by filtering out identical matrices from $\Gc_{m_i}(M)$; formally, we construct a~surjection $\filtering_{m_i}$ mapping $\Gc_{m_i}(M)$ onto the set of objects in $\layer{i}$ such that $\filtering_i(a, b) = \filtering_i(a', b')$ if and only if $\zone{m_i}{a}{b} = \zone{m_i}{a'}{b'}$.
For $i = \ell$, this will be done directly, by listing all entries in the zone; for $i < \ell$, this will be done by taking all representative zones and comparing the subzones in their $m_{i+1}$-regular divisions.
Finally, the pointers from $\layer{i}$ to $\layer{i+1}$ will be derived from the mapping $\mapping_{m_{i+1}}$.

\subparagraph*{The bottom layer $\layer{\ell}$.}
We begin by constructing $\layer{\ell}$.
Using Lemma~\ref{lem:efficient-zone-representatives}, we find $\Gc_{m_\ell}(M)$ and $\mapping_{m_\ell}$.
Next, for each representative $(i, j) \in \Gc_{m_\ell}(M)$, we examine each individual entry in $\zone{s}{i}{j}$ using $m_\ell^2$ queries to $\Ic_L$.
This requires $\Oh[d]{n \log n}$ queries in total for all elements of $\Gc_{m_\ell}(M)$, resulting in time $\Oh[d]{n \log n \log \log n}$.
Thus, each representative zone is now fully described by a~bitvector of length $m_\ell^2 \in \Oh[d]{\log^2 n}$; and two representative zones $\zone{m_\ell}{a}{b}$, $\zone{m_\ell}{a'}{b'}$ are equal if and only if their corresponding bitvectors are equal.
The bitvectors can be sorted using radix sort in time $\Oh[d]{\frac{n}{\log n} \cdot \log^2 n} = \Oh[d]{n \log n}$.
Then, the zones can be grouped into equivalence classes with respect to their equality; each such class corresponds to one zone in the zone family $\Fc_{m_\ell}(M)$.
Eventually, we pick one matrix from each equivalence class and store it in its entirety as an~object of $\layer{\ell}$.

This concludes the construction of $\layer{\ell}$.
The time complexity is $\Oh[d]{n \log n \log \log n}$, dominated by querying $\Ic_L$ for individual elements of $M$.
The choice of a~matrix from each class induces the surjection $\filtering_{m_\ell}$.

\subparagraph*{Layers $\layer{i}$ for $i < \ell$.}
Assume that the layer $\layer{i+1}$ has already been constructed, together with the~auxiliary set $\Gc_{m_{i+1}}(M)$ and functions $\mapping_{m_{i+1}}$ and $\filtering_{m_{i+1}}$.
By Lemma~\ref{lem:efficient-zone-representatives}, we find $\Gc_{m_i}(M)$ and $\mapping_{m_i}$.

We enumerate the objects in $\layer{i + 1}$ as $A_1, A_2, \dots, A_s$, where $s = |\Fc_{m_{i+1}}(M)| \in \Oh[d]{\frac{n}{m_{i+1}}}$.
Recall that each $A_j$ corresponds to a~different matrix in the zone family $\Fc_{m_{i+1}}$.
Then, take some $(a, b) \in \Gc_{m_i}(M)$ and let $Z = \zone{m_i}{a}{b}$ denote the corresponding~representative zone in~$M$.
We list all subzones $\szone{Z}{p}{q}$ ($p, q \in [m_i / m_{i+1}]$) in the regular $m_{i+1}$-division of $Z$ and interpret each of them as an~element of $\Fc_{m_{i+1}}(M)$.
Since $\szone{Z}{p}{q} = \zone{m_{i+1}}{(a - 1)\frac{m_i}{m_{i+1}} + p}{\, (b - 1)\frac{m_i}{m_{i+1}} + q}$, we can find the unique element $j \in [s]$ such that $A_j = \object{\szone{Z}{p}{q}}$ in $\Oh{\log \log n}$ time by locating the representative zone of $\szone{Z}{p}{q}$ in the $m_{i+1}$-regular division of $M$ using $\mapping_{m_{i+1}}$, and then using $\filtering_{m_{i+1}}$ to find the corresponding object in $\layer{i + 1}$.
This way, we describe each representative zone $Z$ of the regular $m_i$-division of $M$ as an~$\frac{m_i}{m_{i+1}} \times \frac{m_i}{m_{i+1}}$ square matrix $\Dc(Z)$ of elements from $1$ to $s$; again, two representative zones $Z_1$, $Z_2$ are equal to each other if and only if their descriptions are equal.
As $|\Gc_{m_i}| \in \Oh[d]{\frac{n}{m_i}}$ and we spend $(m_i / m_{i+1})^2$ calls to $\mapping_{m_{i+1}}$ for each zone in $\Gc_{m_i}$, this in total requires time
$\Oh[d]{\frac{n}{m_i} \cdot \left(\frac{m_i}{m_{i+1}}\right)^2 \cdot \log \log n}$, which by $m_{i+1}^2 \geq m_i$ is bounded by $\Oh[d]{n \log \log n}$.

We now filter the repeated occurrences of the descriptions of the representative zones in $\Gc_{m_i}$.
We do it by sorting the descriptions in the lexicographic row-major order using any comparison sort, where in each comparison we simply compare two arrays of length $(m_i / m_{i+1})^2$, and then grouping equal descriptions. This takes time $\Oh[d]{\frac{n}{m_i} \cdot \log \frac{n}{m_i} \cdot \left(\frac{m_i}{m_{i+1}}\right)^2}$, which is bounded by $\Oh[d]{n \log n}$.
Afterwards, we pick one representative from each equivalence class and store it as~an object in $\layer{i}$.
Again, each object of $\layer{i}$ corresponds to a~single zone in the zone family $\Fc_{m_i}$, and the construction above naturally gives rise to the surjection~$\filtering_{m_i}$.
Given an~object $\object{Z} \in \layer{i}$, the pointers from $\object{Z}$ to the objects of $\layer{i+1}$ can be immediately deduced from $\Dc(Z)$ and the sequence $A_1, A_2, \dots, A_s$.

\subparagraph*{Summary.}
We constructed the bottom-most layer $\layer{\ell}$ in time $\Oh[d]{n \log n \log \log n}$.
For each $i = \ell - 1, \ell - 2, \dots, 0$, the construction of $\layer{i}$ takes time $\Oh[d]{n \log n}$, dominated by the comparison sort of the descriptions of the zones.
Since $\ell \in \Oh{\log \log n}$, we conclude that the time complexity of the entire construction is $\Oh[d]{n \log n \log \log n}$.
Therefore, the constructive part of Theorem~\ref{thm:main} is proved.


\bibliographystyle{plainurl}
\bibliography{references}

\newpage

\appendix

\section{Representation with bitsize $\Oh{n^{1+\varepsilon}}$ and query time $\Oh{1 / \varepsilon}$}
\label{sec:superlinear-repr}

In this section we provide a~brief sketch of another data structure representing twin-ordered matrices. For any fixed $\varepsilon > 0$, we will construct a data structure that represents a given $d$-twin-ordered $n\times n$ matrix $M$ in bitsize $\Oh{n^{1+\varepsilon}}$, and can be queried for entries of $M$ in worst-case time $\Oh{1/\varepsilon}$ per query.
Actually, the data structure solves the \ProbOrthogonalLocation problem within the same space and time bounds, provided that the input is given as a~set of orthogonal rectangles with pairwise disjoint interiors, and with integer coordinates between $0$ and $n$.
As the set of $1$ entries in any $d$-twin-ordered matrix $M$ admits a~rectangle decomposition into $\Oh[d]{n}$ rectangles (Lemma~\ref{lem:rectangles}), this also yields a~data structure representing $M$.

Notably, Chan~\cite{DBLP:journals/talg/Chan13} observed that \ProbOrthogonalLocation can be reduced to the static variant of the \textsc{Predecessor Search} problem, even if the input coordinates are from $0$ to $\Oh{n}$.
P{\v{a}}tra{\c{s}}cu and Thorup proved that each data structure for \textsc{Predecessor Search} with $\Oh{n \log^{\Oh{1}} n}$ bitsize necessarily requires $\Omega(\log \log n)$ query time, even in a~much more powerful cell probe model~\cite{DBLP:conf/stoc/PatrascuT06}.
Therefore, for general \ProbOrthogonalLocation, one cannot expect to achieve constant query time with bitsize significantly smaller than $\Oh{n^{1+\varepsilon}}$.

\subparagraph*{Data structure for disjoint intervals.}
Consider integers $k, h \geq 1$, and let $n = k^h$.
We will first sketch a~data structure that maintains a~set of disjoint integer intervals that are subintervals of $[0, n - 1]$.
The data structure shall allow adding or removing intervals in time $\Oh{kh}$ and querying whether a~point is contained in any interval in time $\Oh{h}$.

Consider a~perfect $k$-ary tree of depth $h$.
The tree has $k^h$ leaves, numbered from $0$ to $n - 1$ according to the pre-order traversal of the tree.
Each internal node at depth $i \in \{0, 1, \dots, h - 1\}$ in the tree corresponds to a~contiguous interval of leaves of length $k^{h-i}$.
Each such interval is called a~\emph{base interval}.
Each internal node contains an~array of $k$ pointers to the children in the tree, allowing access to the $j$-th child in constant time.
Additionally, alongside each node $v$ of the data structure, we store an~additional bit $b_v$, initially set to $0$.

Assume an~interval $[\ell, r]$ is to be inserted to the set.
We traverse the tree recursively, starting from the root, entering only nodes whose base intervals intersect $[\ell, r]$, and cutting the recursion at nodes whose base intervals are entirely within $[\ell, r]$.
It can be shown that the recursion visits at most $\Oh{kh}$ nodes and decomposes $[\ell, r]$ into $\Oh{kh}$ disjoint base intervals.
For each node $v$ corresponding to such a~base interval, we set $b_v \gets 1$.
Removing an~interval from the set is analogous.
Now, to verify whether an~element $y$ belongs to the set, we descend recursively from the root of the tree to the $y$-th leaf of the tree and verify if any of the visited nodes $v$ has $b_v = 1$.
This requires time $\Oh{h}$.

Since each update and query to the data structure is essentially a~recursive search from the root of the tree, the data structure can be made persistent: on each update, we create a~copy of each altered node and each of their ancestors, and we reset the pointers in the copies accordingly.
As $\Oh{kh}$ nodes are updated at each query, and each internal node stores an~array of $\Oh{k}$ pointers, the update time increases to $\Oh{k^2h}$ due to the copying of the nodes; and each update increases the bitsize of the data structure by $\Oh{k^2 h \log n}$.
Thus, after $\Oh[d]{n}$ updates, the bitsize of the data structure is $\Oh[d]{nk^2 h \log n}$.
The query time remains at $\Oh{h}$.

\subparagraph*{Orthogonal point location with small coordinates.}
Fix any $\varepsilon > 0$.
Given a~matrix $M$ of order $n$, we set $h \coloneqq \left\lceil 2/\varepsilon \right\rceil + 1$ and $k \coloneqq \left\lceil n^{1/h} \right\rceil$.
We instantiate a~persistent $k$-ary tree of depth $h$ as above.
We sweep the set of rectangles from the left of the right, maintaining a~vertical sweep line.
The tree maintains an intersection of the sweep line with the union of rectangles as a~set of disjoint intervals contained in $[0, n]$.
Hence, for each rectangle, the tree is updated twice: a~vertical interval is added when the sweep line reaches the left end of the rectangle, and is removed as soon as it reaches the right end of the rectangle.
At each $x$ coordinate, we store the pointer $\mathsf{ver}_x$ to the root of the current version of the tree.
After the preprocessing, for each query $(x, y)$, we fetch the pointer $\mathsf{ver}_x$ and check whether this version of the tree contains $y$ as an~element.

Let us analyze the query time and the bitsize of the data structure.
For convenience, let $\delta \coloneqq 1/h$.
We can see that $0 < \delta < \frac{\varepsilon}{2}$.
Each query is performed in time $\Oh{h} = \Oh{1/\varepsilon}$.
Storing pointers $\mathsf{ver}_x$ requires bitsize $\Oh{n \log n}$.
Since we processed $\Oh[d]{n}$ rectangles, the persistent tree has bitsize $\Oh[d]{nk^2 h \log n} = \Oh[d]{n^{1+2\delta} \log n / \varepsilon} = \Oh[d]{n^{1+\varepsilon}}$.

\end{document}